\RequirePackage{xcolor}
\documentclass[journal,twoside,web]{IEEEcolor}
\usepackage{generic}
\usepackage{cite}
\usepackage{booktabs}
\usepackage{multirow}
\let\labelindent\relax

\usepackage[utf8]{inputenc}
\usepackage{graphics} % for pdf, bitmapped graphics files
\usepackage{amsmath} % assumes amsmath package installed
\usepackage{amssymb}  % assumes amsmath package installed
\usepackage{amsfonts}
\usepackage{comment}
\usepackage{enumitem}
\usepackage{cases,xspace,enumitem}

\usepackage{version}
\includeversion{arxiv}
\excludeversion{tac}

\usepackage{tikz}
\usetikzlibrary{positioning}

\usepackage[prependcaption,colorinlistoftodos]{todonotes}

\definecolor{gnred}{RGB}{255,91,89}

\definecolor{gnred1}{RGB}{71,0,0} % 470000
\definecolor{gnred2}{RGB}{117,0,0} % 750000
\definecolor{gnred3}{RGB}{164,0,0} % a40000
\definecolor{gnred4}{RGB}{211,0,0} % d30000
\definecolor{gnred5}{RGB}{255,0,0} % FF0000
\definecolor{gnred6}{RGB}{255,42,34} % FF2a22
\definecolor{gnred7}{RGB}{255,91,89} % ff5b59 --- favorite

\definecolor{gnblue1}{RGB}{0,36,71}   % 002447
\definecolor{gnblue2}{RGB}{0,60,118}  % 003c76
\definecolor{gnblue3}{RGB}{0,85,164}
\definecolor{gnblue4}{RGB}{0,108,212}
\definecolor{gnblue4}{RGB}{0,108,212}
\definecolor{gnblue5}{RGB}{0,133,255}  % 0085ff
\definecolor{gnblue6}{RGB}{35,156,255} % 239cff
\definecolor{gnblue7}{RGB}{88,177,255} % 58b1ff

\definecolor{gnbrown1}{RGB}{71,27,0}  % 471b00
\definecolor{gnbrown2}{RGB}{117,45,0} % 752d00
\definecolor{gnbrown3}{RGB}{164,62,0} % a43e00
\definecolor{gnbrown4}{RGB}{211,80,0} % d35000
\definecolor{gnbrown5}{RGB}{255,97,0} % ff6100
\definecolor{gnbrown6}{RGB}{255,127,26} % ff7f1a
\definecolor{gnbrown7}{RGB}{255,155,86} % ff9b56

\renewcommand{\theenumi}{(\roman{enumi})}

\newcommand\Item[1][]{%
  \ifx\relax#1\relax  \item \else \item[#1] \fi
  \abovedisplayskip=0pt\abovedisplayshortskip=0pt~\vspace*{-\baselineskip}}

\usepackage{hyperref}
\hypersetup{
 colorlinks=true,
 linkcolor=blue,
 filecolor=magenta,      
 urlcolor=cyan,
 pdftitle={Euclidean_Contractivity_Lure},
 pdfpagemode=FullScreen,
 }

\newcommand{\e}{\operatorname{e}}

\usepackage{amsthm}

\newtheoremstyle{ieeeconf}
 {0pt}   % ABOVESPACE
 {0pt}   % BELOWSPACE
 {\normalfont}  % BODYFONT
 {\parindent}       % INDENT (empty value is the same as 0pt)
 {\itshape} % HEADFONT
 {:}         % HEADPUNCT
 { } % HEADSPACE
 {\thmname{#1} \thmnumber{#2}\thmnote{ (#3)}} % CUSTOM-HEAD-SPEC
\makeatletter
\renewenvironment{proof}[1][\proofname]{\par
  \pushQED{\qed}%
  \normalfont \topsep\z@
  \trivlist
  \item[\hskip2em
        \itshape
 #1\@addpunct{:}]\ignorespaces
}{%
  \popQED\endtrivlist\@endpefalse
}
\makeatother

\theoremstyle{ieeeconf}

\newtheorem{defn}{Definition}
\newtheorem{thm}{Theorem}
\newtheorem{lemma}[thm]{Lemma}
\newtheorem{pro}[thm]{Proposition}
\newtheorem{cor}[thm]{Corollary}
\newtheorem{remark}[thm]{Remark}

\newcommand{\R}{\mathbb{R}}
\newcommand{\real}{\mathbb{R}}

\newcommand{\realnonnegative}{\mathbb{R}_{\geq 0}}

\DeclareMathOperator{\Ker}{\operatorname{Ker}}
\DeclareMathOperator{\Img}{\operatorname{Im}}

\newcommand{\norm}[2]{\|#1\|_{#2}}

\newcommand{\kron}{\operatorname{\otimes}}
\newcommand{\diag}[1]{\operatorname{diag}(#1)}

\newcommand{\subscr}[2]{#1_{\textup{#2}}}

\newcommand{\map}[3]{#1 \colon #2 \rightarrow #3}

\newcommand{\Lip}{\operatorname{\mathsf{Lip}}}

\newcommand{\Faug}{\subscr{F}{aug}}
\newcommand{\vectorized}[1]{\mathbf{vec}(#1)}

\newcommand{\mcM}{\mathcal{M}}
\newcommand{\mcN}{\mathcal{N}}
\newcommand{\mcO}{\mathcal{O}}

\newcommand{\mcR}{\mathcal{R}}
\newcommand{\mcT}{\mathcal{T}}
\newcommand{\mcU}{\mathcal{U}}
\newcommand{\mcW}{\mathcal{W}}
\newcommand{\mcX}{\mathcal{X}}
\newcommand{\mcY}{\mathcal{Y}}

\newcommand{\xstar}{x^{\star}}
\newcommand{\ustar}{u^{\star}}

\newcommand{\uext}{u_{\text{ext}}}

\newcommand{\0}{\mbox{\fontencoding{U}\fontfamily{bbold}\selectfont0}}

\newcommand{\fzero}{\0}

\title{Neural networks for control and learning: \\Contractivity and separation principle
}

\title{Contractivity and a Separation Principle in Neural Networks for Control and Learning}

\title{A Nonlinear Separation Principle: \\
Contraction theory, Neural networks and Learning}

\title{A Nonlinear Separation Principle via Contraction Theory: Applications to Neural Networks, Control, and Learning}

\author{Anand Gokhale, Anton V. Proskurnikov, Yu Kawano, Francesco Bullo \thanks{ This work was in part supported by AFOSR project FA9550-22-1-0059 and by JST FOREST Program Grant Number JPMJFR222E. Anand Gokhale and Francesco Bullo are with
    the Center for Control, Dynamical Systems, and Computation, UC Santa Barbara, Santa Barbara, CA 93106 USA. email:  {\tt\small \{anand\_gokhale,bullo\}@ucsb.edu}. \newline
    Anton V. Proskurnikov is with Department of Electronics and Telecommunications, Politecnico di Torino, Italy, 10129.  email: {\tt\small anton.p.1982@ieee.org} \newline
    Yu Kawano is with the Graduate School of Advanced Science and Engineering, Hiroshima University, Higashi-Hiroshima 739-8527, Japan. email: {\tt\small ykawano@hiroshima-u.ac.jp}  \newline
    Anand Gokhale thanks Alexander Davydov for fruitful discussions.
    }}

\begin{document}

\maketitle

\thispagestyle{empty}
\pagestyle{empty}

\begin{abstract}
 This paper establishes a nonlinear separation principle based on contraction theory and derives sharp stability conditions for recurrent neural networks (RNNs). First, we introduce a nonlinear separation principle that guarantees global exponential stability for the interconnection of a contracting state-feedback controller and a contracting observer, alongside parametric extensions for robustness and equilibrium tracking. Second, we derive sharp linear matrix inequality (LMI) conditions that guarantee the contractivity of both firing rate and Hopfield neural network architectures. We establish structural relationships among these certificates—demonstrating that continuous-time models with monotone non-decreasing activations maximize the admissible weight space—and extend these stability guarantees to interconnected systems and Graph RNNs. Third, we combine our separation principle and LMI framework to solve the output reference tracking problem for RNN-modeled plants. We provide LMI synthesis methods for feedback controllers and observers, and rigorously design a low-gain integral controller to eliminate steady-state error. Finally, we derive an exact, unconstrained algebraic parameterization of our contraction LMIs to design highly expressive implicit neural networks, achieving competitive accuracy and parameter efficiency on standard image classification benchmarks.
\end{abstract}
\section{Introduction}

\paragraph{Motivation}
A central problem in control theory is the design of controllers which utilize sensor data to achieve certain objectives, such as stabilization, reference tracking, or optimal control based on some cost function. In linear systems, a separation principle enables the independent design of observers for state reconstruction, and full state feedback controllers. However, in nonlinear systems, a general counterpart remains elusive. Existing approaches typically rely on time-scale separation and high-gain observer constructions~\cite{ANA-HKK:02, AT-LP:95}, or are restricted to specific system classes such as Lur’e systems and control-affine dynamics~\cite{AS-RJ-AR-LF:08, IRM-JJES:14, MG-VA-ST-DA:23}. These methods often impose restrictive structural assumptions and may exhibit undesirable transient behavior, including peaking and sensitivity to noise and model uncertainty~\cite{FE-HKK:92}.

Here, we bridge this gap by establishing a nonlinear separation principle based on contraction theory~\cite{FB:26-CTDS}. A nonlinear system is said to be contracting if any two trajectories of the system approach each other at an exponential rate. While stronger than Lyapunov stability, contraction provides many properties that are desirable in control design.  Of particular interest in this work are (i) the existence of a unique globally exponentially stable equilibrium, (ii) robustness via incremental ISS bounds to model errors, and (iii) robust equilibrium tracking errors in time varying environments~\cite{AD-VC-AG-GR-FB:23f}.

In parallel, neural networks have become common means of modeling plants  in a data-driven form. Recurrent Neural Networks (RNNs), in particular, have garnered significant interdisciplinary interest. They maintain low computational and memory requirements, making them well-suited for deployment in resource-constrained edge applications and real-time data-driven control~\cite{WDA-ALB-MF:24, WDA-ALB-MF:25}. Continuous-time versions of RNNs naturally mirror biological neural circuits~\cite{CJR-DHJ-RGB-BAO:08} and enable lower power consumption in analog circuit design~\cite{DWT-JJH:84}. Beyond traditional sequence processing, RNNs are also central to implicit deep learning: in architectures such as Deep Equilibrium Models (DEQs)~\cite{EW-JZK:20}, a cascade of layers is replaced by the steady-state equilibrium of an RNN.

Motivated by the requirements of our separation principle, we seek to identify sharp, computationally tractable contraction conditions for RNNs. Such conditions not only expand the applicability of our separation principle, but are also of fundamental importance across the aforementioned domains to guarantee reliable, robust implementations. Recent works~\cite{LEG-FG-BT-AA-AYT:21, WDA-ALB-MF:24} derive stability conditions for globally non-expansive activation functions. However, widely used activation functions—such as ReLU, tanh, and sigmoid—are also monotone nondecreasing. Neglecting this additional structure yields overly conservative guarantees. Accordingly, we seek to identify the broadest class of RNNs that are guaranteed to contract under these practical nonlinearities.

\paragraph{Relevant Literature}
Existing approaches to nonlinear separation principles have relied heavily on time-scale-separation arguments and high-gain observers~\cite{ANA-HKK:02,AT-LP:95}. These methods require the observer dynamics to be sufficiently faster than the plant dynamics, 
which may lead to undesirable transient peaks in the closed-loop trajectories and even finite escape times~\cite{FE-HKK:92}. Under various structural assumptions on the system, separation principles have been provided for Lur'e systems~\cite{AS-RJ-AR-LF:08,MG-VA-ST-DA:23}, control-affine systems with linear inputs~\cite{IRM-JJES:14}, and SSMs~\cite{MZ-VG-AK-ECB-GFT:26}. Our formulation is most directly comparable to the foundational work of~\cite{MV:80} which establishes exponential stability via a Lyapunov-based analysis. By adopting a contraction-theoretic framework, we offer an alternative approach that yields three distinct practical advantages. First, we establish global exponential stability with explicit convergence rates in terms of state and observer errors. Second, the incremental input-to-state stability (iISS) inherent to contracting systems provides native robustness against plant model errors. Finally, contraction yields robust equilibrium tracking bounds in time-varying environments~\cite{AD-VC-AG-GR-FB:23f}; this not only addresses exogenous disturbances but also enables seamless augmentation of the controller with external signals to achieve higher-level goals like reference tracking.

Contraction of continuous-time neural networks in the $\ell_1$ and $\ell_\infty$ norms is presented in~\cite{SJ-AD-AVP-FB:21f}, and the sharpest known conditions in the Euclidean norm are limited to the symmetric case~\cite{VC-AG-AD-GR-FB:23c}. Discrete-time robust RNNs are parameterized in~\cite{MR-RW-IRM:21}.

Control design for discrete time RNNs has been previously studied in the context of asymptotic stability~\cite{AN-MSH-RDB:22} and incremental ISS properties~\cite{WDA-ALB-MF:24}, without closed loop guarantees. Structurally, RNNs can be considered a subclass of Lur'e-type systems, i.e., feedback interconnections of LTI blocks and static nonlinear maps. 
\textcolor{black}{A distinguishing feature of RNNs is that their external inputs and outputs may enter and exit the system through a nonlinearity, unlike the affine inputs and linear outputs usually assumed in standard Lur'e-type models~\cite{VA-ST:19, MG-VA-ST-DA:23}. To control RNNs in the general case (in particular, to design a reference tracking controller) we thus employ}
tools from singular perturbation theory~\cite{JWSP:21}. 

Finally, in implicit deep learning, recent findings show that increasing the number of fixed-point iterations can significantly improve model expressivity~\cite{JL-LD-SO-WY:25}, an effect enabled by the model's local Lipschitz dependence on its input. 
\textcolor{black}{Motivated by this insight, we} 
design input-dependent DEQs that are locally Lipschitz by construction.

\paragraph{Contributions}

\textcolor{black}{This paper establishes a nonlinear separation principle based on contraction and derives sharp contractivity conditions for continuous-time and discrete-time firing-rate neural networks (FRNNs) and Hopfield neural networks (HNNs). We apply these results to controller design for RNN-modeled plants and to implicit-model architectures in machine learning.} Our specific contributions are as follows.

First, we introduce a \emph{nonlinear separation principle} (Theorem~\ref{thm:separation_principle}). This principle states that if a plant can be rendered contracting under full-state feedback and observed through a contracting observer, then the resulting output-feedback closed loop is globally exponentially stable, with explicit estimates of both the observer and state errors. Crucially, this guarantee holds even when the coupled closed-loop system is not itself globally contracting. Building on this framework, we extend the separation principle to parametric systems to establish rigorous robustness guarantees. We derive explicit error bounds for two practical scenarios: robustness against model mismatch when the controller and observer are designed using an estimated parameter, and equilibrium tracking performance when the system is subjected to time-varying parameters.

Second, we derive certificates that guarantee the contraction of FRNNs and HNNs across standard classes of activation functions, summarized in Table~\ref{tab:lmi_conditions}, demonstrate their tightness and establish structural relationships among these conditions, summarized by Theorem~\ref{thm:structural_relationships}. Our analysis shows that the set of weight matrices guaranteeing contraction expands when passing from discrete-time to continuous-time models, and enlarges further when the activations are monotone non-decreasing rather than merely non-expansive, as is often assumed~\cite{LEG-FG-BT-AA-AYT:21}. We also show that our certificates are optimal in the case of symmetric matrices. We establish a necessary condition for our certificate to hold for interconnected RNN systems (Theorem~\ref{thm:network_impossibility}): 
such an interconnection may satisfy our contractivity certificate only if each subsystem is either individually satisfies our certificate or can be made to do so via static output feedback. Given that solving for static output feedback gain is an inherently nonconvex problem, this finding further highlights the utility of our separation principle. We also provide sufficient conditions for the global contractivity of Graph RNNs (Theorem~\ref{thm:ignn_contraction}), demonstrating that identical contracting subsystems preserve stability when coupled over undirected graphs, thereby enabling scalable, decentralized computation for graph neural networks.

Third, we combine the proposed separation principle with our contraction certificates to address the reference-tracking problem for plants modeled by FRNNs. We provide LMI methods to synthesize full-state feedback controllers and state observers that satisfy the assumptions of our separation principle and also characterize necessary and sufficient conditions for the feasibility of these LMIs. Beyond exponential stability, we derive a synthesis procedure for a low-gain integral controller. Treating this integral controller as a time-varying parameter, we use the equilibrium tracking corollary of our separation principle to show the resulting closed-loop system achieves reference tracking with global exponential stability.

Finally, we translate our sharp contractivity conditions into practical tools for machine learning. We derive an exact, unconstrained algebraic parameterization of the most expressive set of synaptic matrices characterized by our contraction LMIs (Theorem~\ref{thm:parameterization}). Motivated by~\cite{JL-LD-SO-WY:25}, we leverage our parameterization to design highly expressive, yet parameter efficient implicit neural networks with input-dependent weights. We demonstrate that the resulting architecture demonstrates competitive accuracy on the MNIST and CIFAR-10 benchmarks.%while using fewer parameters.

Compared to the early version of this work~\cite{AG-AVP-YK-FB:26a}, this extended journal version contains several novel contributions. First, we present a nonlinear separation principle within the framework of contraction theory. Second, we provide an extensive analysis of interconnections of RNNs with specialized results for Graph RNNs. Third, 
the results on reference tracking in~\cite{AG-AVP-YK-FB:26a} are substantially generalized; in particular, our analysis accommodates open-loop unstable plants, which are excluded by the conditions in~\cite{AG-AVP-YK-FB:26a}. Finally, we provide a formal analysis with an explicit local Lipschitz bound for our parameter efficient implicit neural networks.

The rest of this paper is organized as follows: We introduce some preliminaries in Section~\ref{sec:background}. We present our first main result, a nonlinear separation principle in Section~\ref{sec:separation_principle}. Next, we provide an analysis of the conditions for contractivity of FRNNs and HNNs including results on interconnections of RNNs in Section~\ref{sec:analysis_frnn_hnn}. The applications of our theoretical efforts in control design are discussed in Section~\ref{sec:control_design} and in machine learning are discussed in Section~\ref{sec:machine_learning}.

\section{Background}\label{sec:background}

\paragraph*{Notation}
We let $\0_{n \times m}$ be the $n \times m$ all zero matrix, $I_n$ be the $n \times n$ identity matrix. For symmetric $A$, we write $A \succ 0$ (respectively, $A \succeq 0$) if $A$ is positive definite (respectively, semidefinite); $A \succ B$ means $A - B \succ 0$. The opposite relations $\prec, \preceq$ are defined analogously.
Given an $n \times n$ matrix $P \succ 0$, we let $\norm{\cdot}{P}$ denote the $P$-weighted Euclidean norm on $\R^n$, defined by $\norm{x}{P} := \sqrt{x^\top P x}$. Given a norm $\norm{\cdot}{\mcX}$ on $\R^n$, and $\norm{\cdot}{\mcY}$ on $\R^m$, and a matrix $A \in \R^{m \times n}$, we denote the induced matrix norm $\norm{A}{\mcX \to \mcY} = \sup_{\norm{x}{\mcX} = 1} \frac{\norm{Ax}{\mcY}}{\norm{x}{\mcX}}$. Given two normed spaces $(\mcX, \norm{\cdot}{\mcX})$ and  $(\mcY, \norm{\cdot}{\mcY})$, a map $\map{F}{\mcX}{\mcY}$ is Lipschitz with constant $\rho \geq 0$, if for all $x, \tilde x \in \mcX$, $\norm{F(x) - F(\tilde x)}{\mcY} \leq \rho \norm{x - \tilde x}{\mcX}$.
\textcolor{black}{We denote by $\Lip(F)$ the (smallest) Lipschitz constant of $F$. If a map depends on several variables, e.g., $x$ and $u$, we write $\Lip_x(F)$ for the Lipschitz constant with respect to the corresponding variable.}
We let $\diag{A_1, A_2, \cdots, A_n}$ denote the block-diagonal matrix, 
where each $A_i$ is either a matrix or a scalar. For each matrix $A \in \R^{m \times n}$, let $A_{\perp}$ denote any full-column-rank matrix satisfying $\Img A_{\perp} = \Ker(A)$, that is, a matrix whose columns form a basis of the kernel of $A$.
\textcolor{black}{Given a matrix $X\in\R^{n\times m}$, we use $\vectorized{X}\in\R^{nm}$ for its column-wise vectorization.} We denote the upper right Dini derivative of the function $f$ $D^+$
\subsection{Contraction theory}\label{subsect:contraction}

Unless otherwise stated, all vector fields $F(t,x)\in\R^n$, where $t\in\realnonnegative$ and $x\in\R^n$, are supposed to be continuous in $t$ and locally Lipschitz in $x$. 
We say that $F$ is strongly infinitesimally contracting with respect to $\norm{\cdot}{P}$  with  rate $c > 0$ if for all $x, \tilde x\in \R^n$ and $t \geq 0$ the inequality holds
\begin{align*}
    (F(t, x) - F(t, \tilde x))^\top P (x - \tilde x) &\leq -c\norm{x - \tilde x}{P}^2.
\end{align*}
The associated dynamical system $\dot{x} = F(t,x)$ is then said to be contracting in the same sense.
If $x(t)$ and $\tilde x(t)$ are two trajectories of this system, then $\norm{x(t) - \tilde x(t)}{P} \leq \e^{-c(t- t_0)}\norm{x(t_0) - \tilde x(t_0)}{P}$, for all $t \geq t_0 \geq 0$. A discrete-time system $x^+ = F(t,x)$ is strongly contracting with constant $\rho$ if $\Lip_x(F) \leq \rho < 1$. We refer to~\cite{FB:26-CTDS} for a recent review of contraction theory. A strongly contracting time-invariant system with $F(t,x)=F(x)$ always admits a unique and globally exponentially stable equilibrium.

\subsection{Incremental multipliers}

For our analysis, we utilize incremental matrix multipliers constraints~\cite{MF-AR-HH-MM-GJP:19} paired with the S-lemma.

\begin{defn}[Incremental multiplier matrix]

    \textcolor{black}{Consider a function $\map{\Psi}{\R^m}{\R^m}$, and a symmetric matrix $M \in \R^{(2m)\times (2m)}$. The function $\Psi$ is said to admit the incremental multiplier matrix $M$, if, for any $x,\tilde x\in \R^m$}, 
    \begin{align}
    \begin{bmatrix}
        x - \tilde x \\
        \Psi(x) - \Psi(\tilde x)
    \end{bmatrix}^\top
    M
    \begin{bmatrix}
        x - \tilde x \\
        \Psi(x) - \Psi(\tilde x)
    \end{bmatrix}
    \ge 0. \label{eq:imm-d}
\end{align}
\end{defn}

Henceforth, we primarily consider \emph{diagonal} nonlinearities $\Psi:\R^m\to\R^m$, where the coordinate $\Psi_{[i]}$ depends solely on $x_{[i]}$. We call such a map slope-restricted in $[k_1, k_2]$, where $-\infty < k_1 \leq k_2 < +\infty$, if all coordinate functions satisfy 
\begin{align*}
k_1\leq\frac{\Psi_{[i]}(x_{[i]})-\Psi_{[i]}(\tilde x_{[i]})}{x_{[i]}-\tilde x_{[i]}}\leq k_2\quad\forall x_{[i]},\tilde x_{[i]}\in\R,\; x_{[i]}\ne\tilde x_{[i]}.
\end{align*}
In the context of neural networks, we define two classes of diagonal nonlinearities based on slope restrictions.

\begin{defn}[CONE and MONE nonlinearities]\label{def:cone_mone}
    A diagonal nonlinearity $\map{\Psi}{\R^n}{\R^n}$ is said to be:
    \begin{enumerate}
        \item {component-wise non-expansive (CONE)} if it is slope-restricted in $[-1, 1]$, and
        \item {monotonically non-decreasing and non-expansive (MONE)} if it is slope-restricted in $[0, 1]$.
    \end{enumerate}
\end{defn}
%The following lemma is known.
\begin{lemma}[Incremental multiplier matrices, see{~\cite[E3.25]{FB:26-CTDS}}] \label{lem:IMM_conditions}
    If a diagonal nonlinearity $\map{\Psi}{\R^m}{\R^m}$ is slope-restricted in $[k_1, k_2]$, then for each diagonal positive definite $Q \in \R^{m \times m}$, $\Psi$ admits the incremental multiplier matrix:
    \begin{align*}
        M = \begin{bmatrix}
            -2k_1k_2 Q & (k_1 + k_2) Q \\ 
            (k_1 + k_2) Q  & -2Q
        \end{bmatrix}.
    \end{align*}
    Consequently, the multiplier matrices
    \begin{equation*}
        M_{\textup{CONE}} = \begin{bmatrix}
            Q & \fzero_{n \times n} \\
            \fzero_{n \times n} & -Q
        \end{bmatrix} \enspace\text{and}\enspace
        M_{\textup{MONE}} = \begin{bmatrix}
            \fzero_{n \times n} & Q \\
            Q & -2Q
        \end{bmatrix}
    \end{equation*}
    are admitted by CONE and MONE nonlinearities respectively.
\end{lemma}

\subsection{Contractivity of Lur'e-type systems} \label{sec:lure}
We consider both the continuous and discrete time Lur'e-type systems. These are presented below, 
\begin{align}
    \dot x &= Ax + B\Psi(Hx), \quad \textrm{and } \label{eq:cts_lure} \\
    x^+ &= Ax + B\Psi(Hx) \label{eq:disc_lure}
\end{align}
Here $x \in \R^n$, $A \in \R^{n \times n}$, $B \in \R^{n \times m}$, $H \in \R^{m \times n}$ and $\map{\Psi}{\R^m}{\R^m}$ is a nonlinearity. We start with a criteria for absolute contractivity of systems~\eqref{eq:cts_lure} and~\eqref{eq:disc_lure}. The term ``absolute'' emphasizes that the property of contractivity is guaranteed uniformly over a class of nonlinearities, which, throughout this paper, is characterized by a given incremental multiplier $M$. Note that if $\Psi$ admits the incremental multiplier $M$, then $\Psi(Hx)$ admits the incremental multiplier
\begin{equation}\label{eq:imm-d-outp} 
M_H=\begin{bmatrix} H & \0_{m\times{m}} \\ \0_{m\times{n}} & I_m   \end{bmatrix}^{\top}M\begin{bmatrix}H & \0_{m\times{m}} \\ \0_{m\times{n}} & I_m   \end{bmatrix}.   
\end{equation}
{\color{black}
This criteria follows from the trivial direction of the S-lemma, using the matrix inequality characterizing the nonlinearity to imply the matrix inequality for contractivity.} The continuous-time part of the next lemma is presented in~\cite[Theorem 4.2]{LDA-MC:13}, we present an extension to discrete time. 

\begin{lemma}[Absolute contractivity of Lur'e systems] \label{lem:absolute_contractivity_lure}
\textcolor{black}{Let $\Psi$ admit an incremental multiplier $M \in \R^{2m \times 2m}$, and let $M_C$ be defined as in~\eqref{eq:imm-d-outp}. Then, the system~\eqref{eq:cts_lure} is strongly infinitesimally contracting with respect to $\norm{\cdot}{P}$ with rate $c > 0$, where $P=P^{\top}\succ 0$ is a $n\times n$ matrix, if}
    \begin{equation}\label{eq:cts_lure_lmi}
      \begin{bmatrix}
        PA+A^\top P + 2c P & PB \\    B^\top P & \0_{m\times{m}}
      \end{bmatrix} + M_H
      \preceq 0.
    \end{equation}    
The system~\eqref{eq:disc_lure} is strongly contracting with respect to $\norm{\cdot}{P}$ with a constant $\rho \in [0, 1)$ if  
    \begin{equation}\label{eq:disc_lure_lmi}
      \begin{bmatrix}
        A^\top P A - \rho^2 P &  A^\top P B  \\
        B^\top P A &  B^\top P B
      \end{bmatrix} + M_H
      \preceq 0.
    \end{equation}  
\end{lemma}
\begin{proof}
    The proof is presented in Appendix~\ref{app:lure_proof}.
\end{proof}

\section{A Nonlinear Separation Principle}\label{sec:separation_principle}

Certifying the stability of an observer-based feedback loop is a fundamental challenge in nonlinear control. While the classical separation principle elegantly resolves this for linear systems, no general analog exists in the nonlinear setting as
dynamic interconnections of nonlinear systems do not inherently preserve stability properties. Existing nonlinear separation results typically rely on time-scale-separation arguments~\cite{ANA-HKK:02,AT-LP:95} and high-gain observers whose
well-known drawbacks are peaking and, in some cases, finite escape times~\cite{FE-HKK:92}.

Departing from high-gain methods, we adopt a contraction-theoretic framework. We seek to characterize the closed loop stability of the interconnection of an open loop unstable plant with an observer and a state-feedback controller. Standard interconnection theorems in contraction (e.g.,~\cite[Theorem 3.23]{FB:26-CTDS}) require that each subsystem is contracting, precluding their application to open-loop unstable plants. \textcolor{black}{The node-wise contraction assumption is inherent to contraction criteria based on log-norm bounds of the network Jacobian, as implied by~\cite[Theorem 2.32]{FB:26-CTDS}. The following theorem shows, that when an open-loop unstable plant is rendered contracting by a controller and observed through a contracting observer, the closed-loop system -- although not contracting itself -- still preserves a key property of contracting systems: the global exponential stability of the equilibrium.
Denote for brevity
\begin{align}\label{eq:beta_def}
    \beta(t; c_K, c_O) = \begin{cases} \dfrac{e^{-c_O t} - e^{-c_K t}}{c_K - c_O}, & c_K \neq c_O, \\[1ex] te^{-c_K t}, & c_K = c_O, \end{cases}
\end{align}
}

\begin{thm}[Separation principle for contracting controllers and observers] \label{thm:separation_principle}
Given maps $\map{F}{\real^n \times \real^m}{\real^n}$ and $\map{h}{\real^n}{\real^p}$,
consider the control system $\dot{x} = F(x,u)$ with output $y = h(x)$. Let $\map{L}{\real^n \times \real^m \times \real^p}{\real^n}$ be an observer, and let $\map{K}{\real^n}{\real^m}$ be a controller. Given norms $\norm{\cdot}{\mcX}$ and $\norm{\cdot}{\mcO}$ on $\real^n$ and a $\norm{\cdot}{\mcU}$ on $\real^m$, assume the following conditions:
    \begin{enumerate}[label=\textup{(A\arabic*)},leftmargin=*]
        \item\label{aSP:PlantC} (Plant contraction by state feedback) The dynamics $\dot{x} = F(x, K(x))$ is strongly infinitesimally contracting with rate $c_K$ with respect to $\norm{\cdot}{\mcX}$, \textcolor{black}{whose equilibrium is $\xstar$}.
        \item\label{aSP:LipBounds} (Lipschitz controller and plant input) Let $K(x)$ be $\ell_K$-Lipschitz in $x$ from $\norm{\cdot}{\mcO}$ to $\norm{\cdot}{\mcU}$. Let $F(x,u)$ be $\ell_u$-Lipschitz in $u$, from $\norm{\cdot}{\mcU}$ to $\norm{\cdot}{\mcX}$, uniformly in $x$.        
        \item\label{aSP:ZeroObserver}(\textcolor{black}{Plant-observer matching}) %(Zero observer injection) 
        $L(x, u, h(x)) = F(x,u)$ for all $x$ and $u$. 
        \item\label{aSP:ObserverC} (Observer contraction) The dynamics $\dot{z} = L(z, u, y)$ is strongly infinitesimally contracting with respect to $\norm{\cdot}{\mcO}$ with rate $c_O$, uniformly in inputs $u$ and $y$.        
    \end{enumerate}
    Then the closed-loop system
    \begin{subequations}\label{eq:closed_loop_seperation}
        \begin{align}
            \dot{x} &= F(x, K(\xi)), \label{eq: controller_seperation}\\
            \dot{\xi} &= L(\xi, K(\xi), h(x)) \label{eq: observer_seperation}
        \end{align}
    \end{subequations}
    \textcolor{black}{has a globally exponentially stable equilibrium $(\xstar,\xstar)$, and its trajectories satisfy the following error bounds:}
    \begin{align}
      &\norm{\xi(t) - x(t)}{\mcO} \leq \norm{\xi(0) - x(0)}{\mcO} \e^{-c_O t}, \label{eq:obs-error}\\
      &\norm{x(t)-\xstar}{\mcX} \leq \norm{x(0)-\xstar}{\mcX} \e^{-c_K t}\notag \\
         &\quad + \ell_u\ell_K \norm{\xi(0)-x(0)}{\mcO}\,\beta(t;c_K,c_O).\label{eq:contr-error}
    \end{align}
    {\color{black} where $\beta(t;c_K,c_O)$ is defined in~\eqref{eq:beta_def}.}
\end{thm}

\begin{proof}
    Combining~\ref{aSP:PlantC} and~\ref{aSP:ZeroObserver}, we find that $(\xstar,\xstar)$ is an equilibrium of the closed-loop system~\eqref{eq:closed_loop_seperation}, since $L(\xstar, K(\xstar), h(\xstar)) = F(\xstar, K(\xstar)) = \0_n$.
    \textcolor{black}{To establish~\eqref{eq:obs-error}, we observe that $z(t) = x(t)$ and $z(t) = \xi(t)$ are both trajectories of $\dot{z} = L(z, u, y)$ corresponding to the same input signals $u(t) = K(\xi(t))$ and $y(t) = h(x(t))$. For $\xi(t)$, this follows from~\eqref{eq: observer_seperation}; for $x(t)$, it is implied 
    by~\eqref{eq: controller_seperation} and Assumption~\ref{aSP:ZeroObserver}, because $\dot x=F(x,u)=L(x,u,h(x))=L(x,u,y)$. 
    Due to the contractivity Assumption~\ref{aSP:ObserverC}, all solutions of $\dot z = L(z, u, y)$ converge exponentially to one another, resulting in~\eqref{eq:obs-error}.}
    Next, define the augmented plant dynamics $ \dot{x} = F(x, K(x)) + v = \Faug(x, v)$. By~\ref{aSP:PlantC}, $\Faug$ is strongly infinitesimally contracting with respect to $x$, uniformly with respect to exogenous input $v$, and 
    $1$-Lipschitz with respect to $v$, uniformly in $x$. Consider two trajectories of this system: one with initialization $x_0$ and input $v_1 = F(x, K(\xi)) - F(x, K(x))$, and the second the nominal stationary trajectory with initialization $\xstar$ and input $v_2 = 0$. Using the incremental ISS property of contracting systems with Lipschitz inputs~\cite[Theorem~3.16]{FB:26-CTDS} and Assumption~\ref{aSP:LipBounds}, we bound the distance between these trajectories:
    \begin{align}
        D^+ \norm{x - \xstar}{\mcX} &\leq -c_K \norm{x - \xstar}{\mcX} + \norm{v_1 - v_2}{\mcX} \nonumber\\
        &\leq -c_K \norm{x - \xstar}{\mcX} + \ell_u\ell_K\norm{\xi - x}{\mcO}. \label{eq:fb_error} 
    \end{align}    
    Substituting the observer error bound~\eqref{eq:obs-error} into~\eqref{eq:fb_error} yields the linear differential inequality:
    \begin{align*}
        D^+ \norm{x - \xstar}{\mcX} \leq &-c_K \norm{x - \xstar}{\mcX} \\
        &+ \ell_u\ell_K \norm{\xi(0) - x(0)}{\mcO} \e^{-c_O t}.
    \end{align*}
    Applying the Comparison Lemma, we obtain
    \begin{align*}
    \norm{x(t) - &\xstar}{\mcX} \leq \norm{x(0) - \xstar}{\mcX} \e^{-c_K t} \\
    &+ \ell_u\ell_K \norm{\xi(0) - x(0)}{\mcO} \int_0^t \e^{-c_K(t-s)} \e^{-c_O s} ds.
\end{align*}
\textcolor{black}{Recognizing the latter integral as $\beta(t; c_K, c_O)$, we obtain~\eqref{eq:contr-error}. Since all norms on $\R^n$ are equivalent, $(x(t), \xi(t))$ converges to $(\xstar, \xstar)$ for all initial conditions $x(0),\xi(0)\in\R^n$.}
\end{proof}

\color{black}
\begin{remark}[Convergence rate]
The closed-loop convergence rate is determined by $\beta(t; c_K, c_O)$, the slowest term on the right-hand side of~\eqref{eq:obs-error},~\eqref{eq:contr-error}. This function decays as $O(e^{-ct})$ for $c < c_*=\min(c_K, c_O)$, with $c = c_*$ when for $c_K \neq c_O$.
\end{remark}
\color{black}

\begin{remark}[Extension to discrete time and sampled data systems]
    Analogous result holds for discrete-time systems, with parallel assumptions on contraction and Lipschitz bounds.
\end{remark}

\begin{remark}[Relationship to the linear separation principle]
    Our result seeks to be a nonlinear analog of the classic separation principle~\cite[Theorem 16.10]{JPH:09}. Indeed, when $F(x, u) = Ax + Bu$, $h(x) = Cx$, $K(\xi)=K_0\xi$ and the Luenberger observer has gain $L_0$, the assumptions in Theorem~\ref{thm:separation_principle} reduce to stabilizability of $(A,B)$ and detectability of $(A, C)$. The contraction rates $c_K$ and $c_O$ correspond to the spectral abscissa of $A - BK_0$ and $A - L_0C$, respectively. 
\end{remark}

\begin{remark}[Relationship to iISS and iIOSS]
    Assumptions~\ref{aSP:PlantC} and~\ref{aSP:LipBounds} together imply that the plant can be rendered incremental input-to-state stable (iISS) via full state feedback~\cite{FB:26-CTDS}, while Assumptions~\ref{aSP:ZeroObserver} and~\ref{aSP:ObserverC} imply that the plant is incremental input/output-to-state stable (iIOSS)~\cite{EDS-YW:97}.
\end{remark}

\subsubsection*{Parametric extensions}
    Theorem~\ref{thm:separation_principle} also extends to parameterized systems. If the plant, output map, observer, and controller maps further depend on $\theta \in \Theta$ and Assumptions~\ref{aSP:PlantC}--\ref{aSP:ObserverC} hold uniformly over $\Theta$, the system globally exponentially converges to the parameter-dependent fixed point $(x^*(\theta), x^*(\theta))$, with the same bounds as in Theorem~\ref{thm:separation_principle}. 

Extending these results to parametric systems allows one to translate the robustness properties of contracting dynamics into robust stability of
the closed-loop system with explicit convergence rate estimates. 
We formalize this in two corollaries. The first establishes error bounds when the controller and observer design relies on an imprecise parameter's estimate. The second bounds the tracking error \textcolor{black}{relative to the instantaneous equilibrium $\xstar(\theta(t))$ when the parameter varies in time.}

\begin{cor}[Robustness to model error] Let $\Theta$ be a subset of a normed space with norm $\|\cdot\|_{\Theta}$.
Given maps $\map{F}{\real^n \times \real^m \times \Theta}{\real^n}$ and $\map{h}{\real^n \times \Theta}{\real^p}$,
consider the control system $\dot{x} = F(x,u,\theta)$ with output $y = h(x,\theta)$. Let $\map{L}{\real^n \times \real^m \times \real^p \times \Theta}{\real^n}$ be an observer, and let $\map{K}{\real^n \times \Theta}{\real^m}$ be a controller.  Let Assumptions~\ref{aSP:PlantC}--\ref{aSP:ObserverC} hold uniformly for $\theta \in \Theta$. Further, \textcolor{black}{let $K$ be $\ell_{K, \theta}$-Lipschitz in $\theta$, uniformly in $x$, and let $L$ be $\ell_{L, \theta}$-Lipschitz in $\theta$, uniformly in its other arguments.} Consider the closed-loop system where the plant corresponds to the true parameter $\theta$, while the controller and the observer use the estimate $\hat{\theta}$:
\begin{subequations}\label{eq:parametric_closed_loop}
\begin{align}
\dot{x} &= F(x, K(\xi, \hat{\theta}), \theta), \label{eq:parametric_plant}\\
\dot{\xi} &= L(\xi, K(\xi, \hat{\theta}), h(x, \theta), \hat{\theta}). \label{eq:parametric_observer}
\end{align}
\end{subequations}
\textcolor{black}{Then, the solutions satisfy the differential inequalities}
\begin{align} 
D^+ \norm{\xi - x}{\mcO} &\leq -c_O \norm{\xi - x}{\mcO} + \ell_{L,\theta} \norm{\hat{\theta} - \theta}{\Theta}, \label{eq:iss_bound_obs_parametric} \\
D^+ \norm{x - \xstar(\theta)}{\mcX} &\leq -c_K \norm{x - \xstar(\theta)}{\mcX} + \ell_u\ell_K \norm{\xi - x}{\mcO} \nonumber\\
&\quad + \ell_u \ell_{K,\theta} \norm{\hat{\theta} - \theta}{\Theta}. \label{eq:iss_bound_state_parametric}
\end{align}
\end{cor}

\begin{proof}
    \textcolor{black}{Consider the trajectory $(x, \xi)$ of the closed loop system~\eqref{eq:parametric_closed_loop}; we suppress the time $t$ throughout. 
    As in the proof of Theorem~\ref{thm:separation_principle}, $z=x$ and $z=\xi$ are, respectively, the solutions to the dynamical system $\dot z=L(z,u,y,\vartheta)$,
    corresponding to the same pair of signals $u=K(\xi,\hat\theta)$, $y=h(x,\theta)$ yet different parameters $\vartheta=\theta$ and $\vartheta=\hat\theta$; this entails~\eqref{eq:iss_bound_obs_parametric} thanks to the iISS property of contracting systems~\cite[Theorem~3.16]{FB:26-CTDS}. Similarly, $x$ and $\xstar(\theta)$ are two solutions of the system $\dot{x} = F(x, K(x, \theta), \theta) + v$, corresponding to inputs $v_1 = F(x, K(\xi, \hat{\theta}), \theta) - F(x, K(x, \theta), \theta)$ and $v_2=0$, respectively. As in the proof of Theorem~\ref{thm:separation_principle}, the iISS property implies~\eqref{eq:iss_bound_state_parametric}. Compared to~\eqref{eq:fb_error}, the upper bound for $\norm{v_2-v_1}{\mcX}$ acquires an additional term $\ell_{K,\theta}\norm{\theta-\hat\theta}{\Theta}$, since the controller is parameter-dependent.}    
\end{proof}

\textcolor{black}{By Assumption~\ref{aSP:PlantC} and~\cite[Lemma~16]{AD-VC-AG-GR-FB:23f}, the map $\theta \mapsto \xstar(\theta)$ is Lipschitz, so $\xstar(\theta(t))$ is locally Lipschitz and almost everywhere differentiable for any smooth curve $\theta \colon (a,b) \to \Theta$ thanks to\cite[Lemma~17]{AD-VC-AG-GR-FB:23f}. The differential iISS property~\cite[Theorem~3.16]{FB:26-CTDS} extends to inputs defined almost everywhere, and can be turned into the integral form via the standard Gr\"onwall inequality; see~\cite[Theorem~2]{AD-VC-AG-GR-FB:23f} for details.}

\begin{cor}[Equilibrium tracking] \label{cor:separation_moving_target}
Let $\Theta$ be a subset of a normed space with norm $\|\cdot\|_{\Theta}$. Given maps $\map{F}{\real^n \times \real^m \times \Theta}{\real^n}$ and $\map{h}{\real^n \times \Theta}{\real^p}$,
consider the control system $\dot{x} = F(x,u, \theta)$ with output $y = h(x,\theta)$. Let $\map{L}{\real^n \times \real^m \times \real^p \times \Theta}{\real^n}$ be an observer and  $\map{K}{\real^n \times \Theta}{\real^m}$ be a controller.  Let Assumptions~\ref{aSP:PlantC}--\ref{aSP:ObserverC} hold uniformly for $\theta \in \Theta$. Let $K$ be $\ell_{K, \theta}$-Lipschitz in $\theta$, uniformly in $x$, and let $F$ be $\ell_{F, \theta}$-Lipschitz in $\theta$, uniformly in its other arguments. Let the parameter $\theta$ evolve along a continuously differentiable trajectory $\theta(t)$. Consider the closed loop system
\begin{subequations}\label{eq:moving_closed_loop}
\begin{align}
\dot{x} &= F(x, K(\xi, \theta(t)), \theta(t)), \label{eq:moving_controller}\\
\dot{\xi} &= L(\xi, K(\xi, \theta(t)), h(x, \theta(t)), \theta(t)). \label{eq:moving_observer}
\end{align}
\end{subequations}
Let $\xstar(\theta(t))$ be the time-varying equilibrium curve of the ideal nominal system. Then, \textcolor{black}{for almost all $t$, one has}
\begin{multline} \label{eq:iss_bound_moving}
D^+ \norm{x - \xstar(\theta(t))}{\mcX} \leq -c_K \norm{x - \xstar(\theta(t))}{\mcX} \\
+ \frac{\ell_u \ell_{K, \theta} + \ell_{F,\theta}}{c_K}\norm{\dot{\theta}(t)}{\Theta}+ \ell_u\ell_K \norm{\xi - x}{\mcO}.
\end{multline}
\end{cor}
\begin{proof}
Consider the auxiliary dynamics,
\begin{align}\label{eq:moving_aux_dynamics}
    \dot{x} = F(x, K(x, \theta(t)),\theta(t)) + w(t) = T(x, \theta(t), w(t)),
\end{align}
where $\map{T}{\R^n \times \Theta \times \R^n}{\R^n}$ is Lipschitz in $w$ with constant $1$ and  Lipschitz in $x$ uniformly in $\theta,w$ due to Assumption~\ref{aSP:LipBounds}. 
By~\ref{aSP:PlantC}, the system~\eqref{eq:moving_aux_dynamics} is
strongly infinitesimally contracting uniformly in $\theta$ and $w$. 
Both $x(t)$ and $\xstar(\theta(t))$ are solutions of~\eqref{eq:moving_aux_dynamics}, corresponding to inputs $w_1(t) = F(x, K(\xi, \theta(t)), \theta(t)) - F(x, K(x, \theta(t)), \theta(t))$ and $w_2(t) = \dot{x}^\star(\theta(t))$ (defined for almost all $t$). 
Using the iISS property~\cite[Theorem 3.16]{FB:26-CTDS} of contracting systems,
\begin{align*}
    &D^+ \norm{x - \xstar(\theta(t))}{\mcX} \nonumber\\
    &\leq -c_K \norm{x - \xstar(\theta(t))}{\mcX} + \norm{w_1 - w_2}{\mcX} \nonumber \\
    &\leq -c_K \norm{x - \xstar(\theta(t))}{\mcX} + \norm{\dot{x}^\star(\theta(t))}{\mcX} \nonumber\\
    &\quad + \norm{F(x, K(\xi, \theta(t)),\theta(t)) - F(x, K(x, \theta(t)),\theta(t))}{\mcX}.
\end{align*}
In view of~\cite[Lemma~17]{AD-VC-AG-GR-FB:23f}, $\norm{\dot{x}^\star(\theta(t))}{\mcX}$ does not exceed $c_K^{-1}\Lip_\theta (F(x, K(x, \theta), \theta))\norm{\dot\theta(t)}{\Theta}$ for almost all $t$, where the Lipschitz constant does not exceed $\ell_u \ell_{K, \theta} + \ell_{F,\theta}$. Furthermore, $\norm{F(x, K(\xi, \theta(t)),\theta(t)) - F(x, K(x, \theta(t)),\theta(t))}{\mcX}\leq \ell_u\ell_K \norm{\xi - x}{\mcO}$ due to~\ref{aSP:LipBounds}. These estimates entail the desired differential inequality~\eqref{eq:iss_bound_moving} for almost all $t$.
\end{proof}
\begin{remark}[Applications of time varying parameters]
    The extension of our bounds to time varying parameters is useful across many domains. Not only can it be used for modeling plants subject to time varying disturbances, but it is also useful for designing controllers which have additional useful characteristics. Later in this work, we utilize this result to design a low gain reference tracking controller in Theorem~\ref{thm:ref_tracking}, other applications include the design of gradient-based controllers building on the work in~\cite{MC-ED-AB:20}, and the design of controllers with time-varying objectives.
\end{remark}

\section{Contraction of firing rate and Hopfield Neural networks} \label{sec:analysis_frnn_hnn}

We are specifically motivated by the study of systems modeled by  the firing rate neural network (FRNN), and the Hopfield neural network (HNN) in both continuous and discrete time. In this section, we present sharp contraction conditions and related properties for these dynamics. The continuous time FRNN and HNN dynamics are given below:
\begin{align}
    \dot{x} &= -x + \Psi(Wx + Bu); \quad y = Cx+Du \label{eq:cts_firing_rate} \\
    \dot{x} &= -x + W\Psi(x) + Bu; \quad y = C\Psi(x)+Du. \label{eq:cts_hopfield}
\end{align}
Here, the state $x\in \R^n$, input $u \in \R^m$,  output $y \in \R^p$, and the synaptic matrix $W \in \R^{n\times n}$, output matrices $C \in \R^{p \times n}$, $D \in \R^{p \times m}$.  $\Psi(\cdot)$ is a diagonal, slope restricted nonlinearity. Equation~\eqref{eq:cts_firing_rate} represents the FRNN dynamics, and equation~\eqref{eq:cts_hopfield} represents the HNN dynamics. The corresponding discrete time dynamics are
\begin{align}
    x^+ &=\Psi(Wx + Bu); \quad y = Cx+Du, \label{eq:disc_firing_rate}\\
    x^+ &= W\Psi(x) + Bu ; \quad y = C\Psi(x)+Du. \label{eq:disc_hopfield}
\end{align}
In the HNN models~\eqref{eq:cts_hopfield},~\eqref{eq:disc_hopfield}, the state $x$ represents the internal membrane potentials of the neurons. Consequently, the measurable output $y$ is read out as a linear combination of their resulting firing rates, $\Psi(x)$. 
\begin{remark}
The state evolution in the FRNN and HNN dynamics~\eqref{eq:cts_firing_rate}-\eqref{eq:disc_hopfield} is structurally similar to a Lur'e model, which enables the use of Lemma~\ref{lem:absolute_contractivity_lure} for the analysis of contractivity. However, in classical Lur'e systems, the input enters linearly, and the output is a linear function of the state. In FRNNs, the input appears \emph{inside} the nonlinearity, and in HNNs, the output is \emph{a nonlinear function} of the state. This nonlinear coupling precludes the direct application of results developed for standard Lur'e systems~\cite{MG-VA-ST-DA:23}.
\end{remark}

\begin{table*}[t]
    \centering
    \renewcommand{\arraystretch}{1.5}
    \caption{Contractivity Conditions for firing rate and Hopfield Models. See Theorem~\ref{thm:main_fr_h_conditions} for definitions of symbols. \newline The neural network with synaptic matrix $W$ corresponding to each entry is contracting (with rate $c$ or factor $\rho$) if there exist $P\succ0$ and diagonal $Q\succ0$ satisfying the corresponding LMI.}
    \setlength{\tabcolsep}{8pt}
    \begin{tabular}{@{}ll | c l r | c l r @{}}
        \toprule
        \textbf{Architecture} & \textbf{Nonlinearity} & \multicolumn{3}{c|}{\textbf{Discrete Time}} & \multicolumn{3}{c@{}}{\textbf{Continuous Time}} \\
        \midrule
        \multirow{2}{*}{\textbf{Firing Rate}} 
        & \textbf{CONE} $[-1, 1]$ & 
        $\begin{bmatrix}
        - \rho^2 P + W^\top Q W  & \0_{n \times n} \\
        \0_{n \times n} & P - Q
        \end{bmatrix}$ & $\preceq 0$ & (\refstepcounter{equation}\theequation)\label{eq:fr_disc_cone} & 
        $\begin{bmatrix}
        -2(1 - c)P + W^\top Q W & P  \\
        P  & - Q
        \end{bmatrix}$ & $\preceq 0$ & (\refstepcounter{equation}\theequation)\label{eq:fr_cts_cone} \\
        \addlinespace
        & \textbf{MONE} $[0, 1]$ &  
        $\begin{bmatrix}
        - \rho^2 P & W^\top Q \\
        Q W & P - 2Q
        \end{bmatrix}$ & $\preceq 0$ & (\refstepcounter{equation}\theequation)\label{eq:fr_disc_mone} &  
        $\begin{bmatrix}
        -2(1 - c)P & P + W^\top Q \\
        P + QW & -2Q
        \end{bmatrix}$ & $\preceq 0$ & (\refstepcounter{equation}\theequation)\label{eq:fr_cts_mone} \\
        \midrule
        \multirow{2}{*}{\textbf{Hopfield}} 
        & \textbf{CONE} $[-1, 1]$ & 
        $\begin{bmatrix}
        - \rho^2 P + Q & \0_{n \times n} \\
        \0_{n \times n} & W^\top P W - Q
        \end{bmatrix}$ & $\preceq 0$ & (\refstepcounter{equation}\theequation)\label{eq:h_disc_cone} & 
        $\begin{bmatrix}
        -2(1 - c)P + Q & PW  \\
        W^\top P  & -Q
        \end{bmatrix}$ & $\preceq 0$ & (\refstepcounter{equation}\theequation)\label{eq:h_cts_cone} \\
        \addlinespace
        & \textbf{MONE} $[0, 1]$ &              
        $\begin{bmatrix}
        - \rho^2 P & Q \\
        Q & W^\top P W - 2Q
        \end{bmatrix}$ & $\preceq 0$ & (\refstepcounter{equation}\theequation)\label{eq:h_disc_mone} & 
        $\begin{bmatrix}
        -2(1 - c)P & PW +  Q \\
        W^\top P + Q & -2Q
        \end{bmatrix}$ & $\preceq 0$ & (\refstepcounter{equation}\theequation)\label{eq:h_cts_mone} \\
        \bottomrule
    \end{tabular}
    \label{tab:lmi_conditions}
\end{table*}

\begin{thm}[Contractivity of FRNN and HNN]\label{thm:main_fr_h_conditions}
    Consider the FRNN and HNN dynamics~\eqref{eq:cts_firing_rate}-\eqref{eq:disc_hopfield}, for some constant $u$. Let $W$ be the synaptic weight matrix and let the activation function $\Psi(\cdot)$ belong to either the CONE or MONE class. \textcolor{black}{If the corresponding matrix inequality in Table~\ref{tab:lmi_conditions} is satisfied by $P\succ 0$, diagonal $Q\succ 0$ and scalar $c>0$ or $\rho\in [0,1)$, then}
    the system is infinitesimally strongly contracting  with rate $c$ (in continuous time) or strongly contracting with factor $\rho$ (in discrete time) in the norm $\norm{\cdot}{P}$. 
\end{thm}

\begin{proof}
     The proof of each of the eight cases follows from  Lemma~\ref{lem:IMM_conditions} and~\ref{lem:absolute_contractivity_lure}. For simplicity, we provide the proof for only the FRNN case in continuous time for MONE nonlinearities; the other cases follow similarly. Since~\eqref{eq:cts_firing_rate} is a continuous-time Lur'e system of the form~\eqref{eq:cts_lure}, with $A = - I_n, B = I_n$ and $H = W$. Per Lemma~\ref{lem:IMM_conditions} a MONE nonlinearity admits the matrix multiplier $M = \begin{bmatrix}
         \0_{n \times n} & Q  \\ Q & -2Q
     \end{bmatrix}$, for some positive diagonal $Q$. Substituting these equalities into~\eqref{eq:cts_lure_lmi} yields~\eqref{eq:fr_cts_mone}.
\end{proof}
We refer to each entry in Table~\ref{tab:lmi_conditions} as a contraction certificate. 
\subsection{Structural relationships of contractivity conditions}
First we show that the discrete-time CONE certificates for both network architectures are reducible to Schur diagonal stability, generalizing prior results on firing-rate models~\cite{WDA-ALB-MF:24}.

\begin{lemma}[Schur diagonal stability]
\label{lem:discrete_cone_schur} 
$W \in \R^{n \times n}$ satisfies the discrete-time CONE certificates~\eqref{eq:fr_disc_cone} and~\eqref{eq:h_disc_cone} if and only if $W$ is Schur diagonally stable.
\end{lemma}
\begin{proof}
\color{black}
    Note that ~\eqref{eq:fr_disc_cone} $\iff P \preceq Q$ and $W^\top QW \preceq \rho^2 P$. Similarly,~\eqref{eq:h_disc_cone} $\iff Q \preceq \rho^2P$ and $W^\top PW \preceq Q$.
    \textcolor{black}{Hence, if~\eqref{eq:fr_disc_cone} holds, then $W^{\top}QW-Q\preceq\rho^2 P-Q\preceq (\rho^2-1)Q\prec 0$, so $W$ is diagonally Schur stable with diagonal matrix $Q$. Analogously, ~\eqref{eq:h_disc_cone} entails diagonal Schur stability, since that $W^{\top}QW-Q\preceq \rho^2W^\top PW-Q\leq (\rho^2-1)Q\prec 0$.} 
        
    Conversely, let $W^{\top}QW-Q\prec 0$ for some diagonal matrix $Q\succ 0$. Then $\exists \; \rho\in(0,1)$ such that $W^{\top}QW-\rho^2Q\preceq 0$. Then, $P=Q$ satisfies~\eqref{eq:fr_disc_cone} and $P=\rho^{-2}Q$ satisfies~\eqref{eq:h_disc_cone}.
\end{proof}

Next, we establish the relationships between the various certificates in Table~\ref{tab:lmi_conditions}.

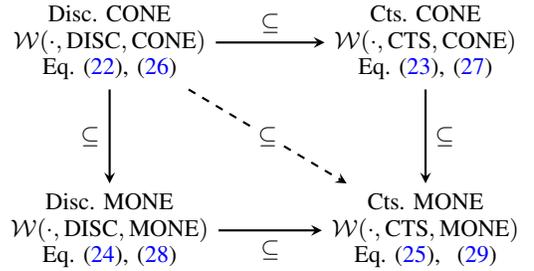
\begin{figure}[htpb]
    \centering

\begin{tikzpicture}[
    >=stealth,
    thick,
    every node/.style={align=center, font=\small}
]
    % Nodes defining the four Firing Rate conditions
    \node (DC) at (0, 2.5) {Disc. CONE \\ $\mathcal{W}(\cdot, \text{DISC}, \text{CONE})$ \\ Eq.~\eqref{eq:fr_disc_cone},~\eqref{eq:h_disc_cone}};
    \node (DM) at (4.2, 2.5) {
    Cts. CONE \\ $\mathcal{W}(\cdot, \text{CTS}, \text{CONE})$ \\ Eq.~\eqref{eq:fr_cts_cone},~\eqref{eq:h_cts_cone}
    };
    \node (CC) at (0, 0) {
        Disc. MONE \\ $\mathcal{W}(\cdot, \text{DISC}, \text{MONE})$ \\ Eq.~\eqref{eq:fr_disc_mone},~\eqref{eq:h_disc_mone}
    };
    \node (CM) at (4.2, 0) {Cts. MONE \\ $\mathcal{W}(\cdot, \text{CTS}, \text{MONE})$ \\ Eq.~\eqref{eq:fr_cts_mone}, ~\eqref{eq:h_cts_mone}};

    % Subset Inclusion Arrows
    \draw[->] (DC) -- node[above] {$\subseteq$} (DM);
    \draw[->] (DC) -- node[left] {$\subseteq$} (CC);
    \draw[->] (CC) -- node[below] {$\subseteq$} (CM);
    \draw[->] (DM) -- node[right] {$\subseteq$} (CM);
    
    \draw[->, dashed] (DC) -- node[fill=white, inner sep=1pt] {$\subseteq$} (CM);
\end{tikzpicture}
\caption{A summary of relationships for the contractivity conditions from Table~\ref{tab:lmi_conditions}. The sets $\mcW(\cdot, \cdot, \cdot)$ are described in Theorem~\ref{thm:structural_relationships}. The discrete time CONE condition restricts the weight matrices the most, whereas the continuous time MONE condition enables maximum expressivity.}
\end{figure}

\begin{thm}[Reductions and duality of the certificates]\label{thm:structural_relationships}
Let $\mcW(\mcM, \mcT, \mcN)$ denote the set of synaptic matrices~$W$ satisfying the contraction certificates in Table~\ref{tab:lmi_conditions} for model $\mathcal{M} \in \{\text{FR}, \text{Hop}\}$, time domain $\mathcal{T} \in \{\text{CTS}, \text{DISC}\}$, and nonlinearity class $\mathcal{N} \in \{\text{CONE}, \text{MONE}\}$. The following hold:
    \begin{enumerate}
        \item The set of synaptic matrices satisfying the condition for CONE nonlinearities is contained in the corresponding set for MONE nonlinearities.\label{item:nonlinearity_inclusion}
        \begin{equation*}
            \mathcal{W}(\cdot, \cdot, \text{CONE}) \subseteq \mathcal{W}(\cdot, \cdot, \text{MONE}).
        \end{equation*}
        
        \item \label{item:discretization} If a synaptic matrix $W$ satisfies a discrete-time condition with factor $\rho \in [0, 1)$, then it also satisfies the corresponding continuous-time condition with rate $c = (1 - \rho^2)/2$:
        \begin{equation*}
             \mathcal{W}(\cdot, \text{DISC}, \cdot) \subseteq \mathcal{W}(\cdot, \text{CTS}, \cdot).
        \end{equation*}
        
        \item \label{item:duality}
        A matrix $W$ satisfies the firing rate condition for some parameters $P \succ 0$ and diagonal $Q \succ 0$ if and only if $W^\top$ satisfies the corresponding Hopfield LMI condition for some $P' \succ 0$ and diagonal $Q' \succ 0$.
        \begin{equation*}
             W \in \mathcal{W}(\text{FR}, \cdot, \cdot) \iff W^\top \in \mathcal{W}(\text{Hop}, \cdot, \cdot).
        \end{equation*}
    \end{enumerate}
\end{thm}

\begin{proof}
    To prove~\ref{item:nonlinearity_inclusion}, notice that the left-hand side of each CONE inequality is the sum of the left-hand side of the corresponding MONE inequality and a negative semidefinite matrix: $M_1 = \begin{bmatrix}
        - W^\top Q W & W^\top Q \\ Q W & - Q
    \end{bmatrix}$ in the FRNN case, and $M_2 = \begin{bmatrix}
        - Q & Q \\ Q & - Q
    \end{bmatrix}$ in the HNN case. Therefore, the CONE conditions imply the corresponding MONE conditions.
    
   To prove~\ref{item:discretization}, notice that each discrete-time LMI condition entails the corresponding continuous-time condition with $c = (1-\rho^2)/2>0$ and $2(1-c)=1+\rho^2$. For instance, in the firing-rate MONE case, the LMI~\eqref{eq:fr_disc_mone} implies~\eqref{eq:fr_cts_mone}:
     \begin{align*}
        &\begin{bmatrix} P & -P \\ -P & P \end{bmatrix} \succeq 0 \implies \nonumber \\
        &\begin{bmatrix} -(1 + \rho^2)P & P + W^\top Q \\ P + QW & -2Q \end{bmatrix} \preceq 
        \begin{bmatrix} -\rho^2P & W^\top Q \\ QW & P -2Q \end{bmatrix}\preceq 0.
    \end{align*} 
    \textcolor{black}{The proofs of~\ref{item:duality} differ in each of the four cases. In the continuous-time cases, under the substitution $(P, Q, W) \mapsto (P^{-1}, Q^{-1}, W^\top)$, each firing-rate condition becomes the corresponding Hopfield condition with $W$ replaced by $W^\top$.}
    For the MONE condition~\eqref{eq:fr_cts_mone}, this is proven via the congruence transformation with matrix
    $T = \diag{P^{-1}, Q^{-1}}=T^{\top}\succ 0$: 
    \begin{equation*}
         T \begin{bmatrix} -2(1 - c)P & \star \\ P + QW & -2Q \end{bmatrix} T = 
         \begin{bmatrix} -2(1 - c)P^{-1} & \star \\ W P^{-1} + Q^{-1} & -2Q^{-1} \end{bmatrix}\preceq 0,
    \end{equation*}
     is equivalent to~\eqref{eq:fr_cts_mone} while also being identical to the LMI~\eqref{eq:h_cts_mone} with new variables $W' = W^\top$, $P' = P^{-1}\succ 0$, and $Q' = Q^{-1}\succ 0$. For the continuous-time CONE case, this is proven via the Schur complement: for $Q\succ 0$, the firing-rate inequality~\eqref{eq:fr_cts_cone} reduces to
    \begin{align}
        -2(1 - c)P + W^\top Q W + PQ^{-1}P \preceq 0.
    \end{align}
    Pre- and post-multiplying this inequality by $P^{-1}$ yields
    \begin{align}\label{eq:cone_duality_step}
        -2(1 - c)P^{-1} + P^{-1}W^\top Q WP^{-1} + Q^{-1} \preceq 0,
    \end{align}
    which is equivalent, by the Schur complement, to the inequality~\eqref{eq:h_cts_cone} with $W' = W^\top$, $P' = P^{-1}$, and $Q' = Q^{-1}$.  
    
    A similar trick, based on the Schur complement, applies to the discrete-time MONE condition: upon rewriting it as
    \begin{align}
        P - 2Q + \rho^{-2} QWP^{-1}W^\top Q \preceq 0
    \end{align}
    and pre- and post-multiplying by $Q^{-1}$, one obtains
    \begin{align}\label{eq:mone_duality_step}
        Q^{-1}PQ^{-1} - 2Q^{-1} + \rho^{-2} WP^{-1}W^\top  \preceq 0,        
    \end{align}
    which is equivalent, by the Schur complement, to the Hopfield matrix inequality~\eqref{eq:h_disc_mone} with parameters $W' = W^\top$, $P' = \rho^{-2}P^{-1}$, and $Q' = Q^{-1}$. Finally, for the discrete time CONE case, both conditions are reducible to Schur diagonal stability as shown in Lemma~\ref{lem:discrete_cone_schur}. 
    We show that $W^\top QW \preceq \rho^2 Q$ if and only if $WQ^{-1}W^\top \preceq \rho^2 Q^{-1}$.
    \textcolor{black}{The former inequality can be written as $X^{\top}X\leq I$, where $X=\rho^{-1}Q^{1/2}WQ^{-1/2}$, and the latter as $XX^{\top}\leq I$. The equivalence follows since $X$ and $X^{\top}$ have same singular values. A more direct proof can be obtained by applying Schur complement to the LMI $\left[\begin{smallmatrix} I & X \\ X^{\top} & I \end{smallmatrix}\right]\succeq 0$ with respect to each diagonal block.}    
\end{proof}

\begin{cor}[Equivalence to LMI]
    For a fixed $c > 0$ or $\rho \in [0,1)$, each certificate in Table~\ref{tab:lmi_conditions} is equivalent to a Linear Matrix Inequality (LMI) under a bijective change of variables. Consequently, for fixed rate, the set of admissible triplets $(W,P,Q)$ is characterized by a convex feasibility problem. 
\end{cor}

\begin{proof}
\color{black}
By the duality between FRNNs and HNNs (Theorem~\ref{thm:structural_relationships}), it suffices to establish the LMI equivalence for one architecture; the other follows via the transformation $(P, Q, W) \mapsto (P^{-1}, Q^{-1}, W^\top)$. For continuous-time Hopfield networks, the substitution $S = W^\top P$ renders both the CONE and MONE conditions linear in $(P, Q, S)$. In discrete time, the firing-rate MONE condition is linear in $(P, Q, S)$ under $S = QW$. Finally, the discrete-time CONE conditions reduce to Schur diagonal stability (Lemma~\ref{lem:discrete_cone_schur}), which admits the LMI reformulation $\left[\begin{smallmatrix} \rho^2 Q & S^\top \\ S & Q \end{smallmatrix}\right] \succeq 0$ via the Schur complement with $S = QW$.
\end{proof}

\subsection{Necessary vs. sufficient conditions for contractivity}

Theorem~\ref{thm:structural_relationships} implies that continuous-time FRNNs and HNNs with MONE nonlinearities admit the largest set of synaptic matrices. Next, we seek to characterize the sharpness of these certificates. We focus on the FRNN case, as HNN results hold via duality. First we relate this certificate to the Lyapunov Diagonal Stability (LDS) of $W{-}I_n$.
\begin{pro}[Relationship with Lyapunov Diagonal Stability]\label{prop:LDS}
    Let $W \in \R^{n\times n}$ be a synaptic matrix for the firing rate neural network~\eqref{eq:cts_firing_rate}. The following hold
    \begin{enumerate}
        \item If $W$ satisfies~\eqref{eq:fr_cts_mone} for $P \succ 0$, diagonal $Q \succ 0$ and $c > 0$, then $W - I_n$ is $Q-$LDS. \label{item:lds_i}
        \item The converse is not true; $W - I_n$ being LDS is not sufficient to guarantee contraction in any fixed norm. \label{item:lds_ii}
    \end{enumerate}
\end{pro}
\begin{proof}
    To prove~\ref{item:lds_i}, we left and right multiply the matrix inequality~\eqref{eq:fr_cts_mone} by the matrix $[I_n, I_n] \in \R^{n \times 2n}$ obtaining,
    \begin{align}
    -2(1-c) P + P + W^\top Q + P + QW - 2Q &\preceq 0 \\
    \implies W^\top Q + QW  \preceq 2Q - 2cP &\prec 2Q.
    \end{align}
    This proves that the $W - I_n$ is $Q-$LDS. 
    
    The claim in~\ref{item:lds_ii} is established by the following counterexample\footnote{This counterexample can be found in the literature, e.g.,~\cite[Theorem 7]{LK-ME-JJES:22}, but we include a proof for completeness.}. Consider a skew-synaptic weight matrix $ W = \begin{bmatrix} 0 & 4 \\ -4 & 0 \end{bmatrix}$. Since $W + W^\top  = 0 \prec 2 I$, $W-I_n$  is LDS. For contraction in the norm $\norm{\cdot}{P}$, the matrix $\mcM(D) = 2P - PDW - W^\top DP $ must be positive definite for all diagonal $D \in [0,I]$, e.g., see~\cite{VC-AG-AD-GR-FB:23c}.  Testing the vertices $D_1 = \diag{1, 0}$ and $D_2 = \diag{0, 1}$,
    \begin{align}
        \mcM(D_1) &= \begin{bmatrix}
            2 p_{11} & 2p_{12} - 4 p_{11} \\
            2p_{12} - 4 p_{11} & 2p_{22} - 8p_{12} 
        \end{bmatrix},  \\
        \mcM(D_2) &= \begin{bmatrix}
            2p_{11} - 8 p_{12} & 2p_{12} - 4 p_{22} \\
            2p_{12} - 4 p_{22} & 2 p_{22}
        \end{bmatrix} . 
    \end{align}
    Since the determinants of these matrices need to be positive, 
    \begin{itemize}
        \item $p_{11} p_{22} > p_{12}^2 + 4p_{11}^2 \implies p_{22} > 4p_{11}$.
        \item $p_{11} p_{22} > p_{12}^2 + 4p_{22}^2 \implies p_{11} > 4p_{22}$.
    \end{itemize}
    These conditions contradict to the condition $P \succ 0$. 
\end{proof}

The next proposition shows that the FRNN MONE certificate is tight for symmetric matrices, recovering the
optimal contraction rate and contractivity characterizations from~\cite{VC-AG-AD-GR-FB:23c}. \textcolor{black}{Note that, in the case of $\Psi(x)=x$, the FRNN and HNN linear dynamics can only be contracting (i.e., stable) when $W-I$ is Hurwitz, that is, the spectral abscissa of $W$ is less than $1$. For $W=W^{\top}$, this condition is also sufficient.}
\begin{pro}[Log-optimal characterization for symmetric weights]\label{prop:symm_matrix}
    Let $W = W^\top \in \R^{n \times n}$ with spectral abscissa $\alpha$. 
    \begin{enumerate}
        \item For $\alpha\leq 0$, the FRNN MONE certificate~\eqref{eq:fr_cts_mone} holds for ${P = - W}$, ${Q = I_n}$, and ${c = 1}$. 
        \item For $0 < \alpha < 1$, there exists $P \succ 0$ such that $ W = P ^{1/2} - (4\alpha)^{-1}P$, and the FRNN MONE certificate~\eqref{eq:fr_cts_mone} holds for this  $P$, ${Q = 4\alpha I_n}$, and $c = 1 - \alpha(W)$. 
    \end{enumerate}
\end{pro}
\begin{proof}
    Consider the Schur complement of~\eqref{eq:fr_cts_mone}. 
    \begin{align}
        -2(1 - c)P + \frac{1}{2}(P + W^\top Q)Q^{-1}(P + QW) \preceq 0.\label{eq:schur_fr_cts_mone}
    \end{align}
If $\alpha<0$, this inequality is satisfied by $P = - W\succ 0$, $Q = I_n$ and $c = 1$ into~\eqref{eq:schur_fr_cts_mone}. If $\alpha\in (0,1)$, it can be shown~\cite[Lemma 10]{VC-AG-AD-GR-FB:23c} that the equation ${W = P^{1/2} - \frac{1}{4\alpha}P}$ has a solution $P \succ 0$.  Choosing $Q = 4\alpha I$, and substituting $W$ and $Q$ into the LHS of~\eqref{eq:schur_fr_cts_mone} and $c=1-a$ yields ${-2(1 - c)P + 2\alpha P}=0$.
\end{proof}
\begin{remark}
    The equivalent results to Proposition~\ref{prop:symm_matrix} for HNNs hold via duality (Theorem~\ref{thm:structural_relationships}).
\end{remark}

Having shown the sharpness of the FRNN and HNN MONE certificates, and their optimality for the class of symmetric matrices, we define the notion of S-contraction corresponding to these certificates.

\begin{defn}[S-contraction] \label{def:s_contraction}
An FRNN (or HNN) with a MONE nonlinearity is said to be \emph{S-contracting} if its synaptic matrix satisfies the FRNN(or HNN) MONE contraction certificate~\eqref{eq:fr_cts_mone} (or~\eqref{eq:h_cts_mone}).
\end{defn}
\begin{remark}[Sharpness and Complexity]
The term \emph{S-contraction} signifies the use of the \emph{S-procedure} to incorporate incremental multiplier matrices. This characterization is:
\begin{enumerate}
\item \emph{Polynomial-time verifiable:} While certifying absolute contractivity for neural networks involves checking $2^n$ matrix inequalities~\cite{VC-AG-AD-GR-FB:23c}, S-contraction is a convex feasibility problem solvable in polynomial time via LMIs. 
\item \emph{Sharp:} The conditions are \emph{sharp} as they constitute the tightest constraints obtainable via the S-lemma using incremental multipliers for absolute contractivity.
\end{enumerate}
\end{remark}

\subsection{Interconnections of RNNs}
The study of interconnected recurrent neural networks (RNNs) is motivated by both biological~\cite{LK-ME-JJES:22} and engineering~\cite{WDA-ALB-MF:24} contexts. We first establish a necessary condition for a network of interconnected RNNs to be S-contracting.
\begin{thm}[Necessary conditions for S-contraction of interconnections of FRNNs]\label{thm:network_impossibility}
Consider an interconnection of $N$ FRNN subsystems~\eqref{eq:cts_firing_rate} with MONE nonlinearities where the $i$th subsystem is given by
\begin{align}
    \dot{x}_i &= - x_i + \Psi_i (W_i x_i + B_i u_i), \\
    y_i &= C_ix_i + D_i u_i,
\end{align}
where $x_i \in \R^{n_i}, u_i \in \R^{m_i}, y_i \in \R^{k_i}$. Consider an interconnection of these systems of the form $    u_i = \sum\nolimits_{j = 1}^N A_{ij}y_j$. Let $\mathbf{A}$ be the block matrix whose $(i,j)$th block is $A_{ij}$. Further, let $W = \diag{W_1, \dots, W_N}$, $B = \diag{B_1, \dots, B_N}$, $C = \diag{C_1, \dots, C_N}$, and $D = \diag{D_1, \dots, D_N}$.

If the interconnection is well-posed, meaning $I_m - \mathbf{A}D$ is invertible (where $m = \sum_{i =1}^N m_i$), then, for ${\Delta=(I_m - \mathbf{A}D)^{-1}\mathbf{A}}$,  
\begin{enumerate}
    \item The interconnected system is an FRNN with synaptic matrix $W_{net} = W + B\Delta C$. \label{item:interconnected_FRNN}
    \item Further, if the interconnected FRNN is S-contracting, then for each $i$, the FRNN with the synaptic matrix $W_i + B_i \Delta_{ii} C_i$ is S-contracting, where $\Delta_{ii}$ is the $i$th principal diagonal block of $\Delta$.\label{item:necessary_condition}
\end{enumerate}
\end{thm}
\begin{proof}
    Let $x = [x_1^\top, \dots, x_N^\top]^\top$, $u = [u_1^\top, \dots, u_N^\top]^\top$, and $y = [y_1^\top, \dots, y_N^\top]^\top$ denote the stacked network vectors. The decoupled open-loop dynamics may be written as
    \begin{align}
        \dot{x} &= -x + \Psi(Wx + Bu)\label{eq:netw-state} \\
        y &= Cx + Du, \label{eq:netw-out}
    \end{align}
    where  $\Psi(z) = [\Psi_1(z_1)^\top, \dots, \Psi_N(z_N)^\top]^\top$. The interconnection imposes the relationship $u = \mathbf{A}y$. 
    When $I_m - \mathbf{A}D$ is invertible, substituting the output from~\eqref{eq:netw-out} yields
    \begin{align}
        u = (I_m - \mathbf{A}D)^{-1} \mathbf{A} C x.
    \end{align}
    Substituting this into~\eqref{eq:netw-state} and denoting $\Delta =(I_m - \mathbf{A}D)^{-1} \mathbf{A}$, the system turns into the FRNN with synaptic matrix $W_{net}$:
\begin{align}\label{eq:interconnected_FRNN}
        \dot{x} &= -x + \Psi((W + B\Delta C)x),
    \end{align}
   which proves~\ref{item:interconnected_FRNN}. Now, if the system~\eqref{eq:interconnected_FRNN} is S-contracting, 
    \begin{align} 
        \begin{bmatrix} -2(1-c)P & P + (W_{net})^\top Q \\ P + Q W_{net} & -2Q \end{bmatrix} \preceq 0.
    \end{align}
    Any principal submatrix of a negative semidefinite matrix must also be negative semidefinite. Extracting the $i$th diagonal block from each of the four entries results in the inequality
    \begin{align} 
        \begin{bmatrix} -2(1-c)P_{ii} & P_{ii} + (W_{net})_{ii}^\top Q_i \\ P_{ii} + Q_i (W_{net})_{ii} & -2Q_i \end{bmatrix} \preceq 0.
    \end{align}
    Note that $(W_{net})_{ii} = W_i + B_i \Delta_{ii} C_i$, where $\Delta_{ii}$ is the $i$th diagonal block of $\Delta$. Therefore, an FRNN with  synaptic matrix $W_i + B_i \Delta_{ii} C_i$ is S-contracting, with $P = P_{ii}, Q = Q_i$. 
\end{proof}

\begin{remark}
    Theorem~\ref{thm:network_impossibility} shows that an interconnection of FRNNs can be S-contracting only if each FRNN is S-contracting either independently or via static output feedback.  When there is no feed-through ($D_i = 0$) and no self-loops ($A_{ii} = 0$), output feedback is unavailable, so each FRNN must be independently S-contracting. 
    Parallel results hold for the interconnected HNNs~\eqref{eq:cts_hopfield}, as well as in discrete time.
\end{remark}

\begin{remark}
     Interconnections of RNNs are particularly useful in control design~\cite{WDA-ALB-MF:24}. \textcolor{black}{Theorem~\ref{thm:network_impossibility} shows that a static output interconnection involving open-loop unstable subsystems can only be contracting if every such subsystem is stabilized by static output feedback, which is an inherently nonconvex problem. This highlights the practical utility of our separation principle (Theorem~\ref{thm:separation_principle}), which relies on a dynamic observer-based controller that can  be designed via  LMI feasibility.} 
\end{remark}

\subsubsection*{Homogeneous networks}
Next, we discuss the special case where the subsystems are identical. Such interconnections arise in graph neural network architectures, such as implicit graph neural networks~\cite{FG-HC-WZ-SS-LEG:20}, continuous graph neural ODEs~\cite{LPX-MQ-JT:20}, and single-layer graph convolutional ODEs~\cite{MP-SM-JP-AY-HA-JP:19}. Consider the Graph FRNN:
\begin{align}\label{eq:cts_graph_implicit_dynamics}
    \dot{X} &= - X + \Psi(WXA + BU). 
\end{align}
Here $X \in \R^{m \times n}$ represents the state matrix, and $U \in \R^{p \times n}$ represents the input node features, $W \in \R^{m \times m}$ is the synaptic matrix of each FRNN, $B \in \R^{m \times p}$ is the input matrix, and $A \in \R^{n \times n}$ is the adjacency matrix. Although the graph FRNN~\eqref{eq:cts_graph_implicit_dynamics} is expressed in matrix form, its vectorized form is structurally identical to a standard FRNN with the synaptic matrix $A^{\top}\kron W$: introducing the state $\tilde X=\vectorized{X}$, one has
\begin{align}
    \dot{\tilde X}  &=  - \tilde{X} + \Psi((A^\top \kron W)\tilde{X} + \vectorized{BU}). \label{eq:cts_graph_implicit_dynamics_vectorized}
\end{align}
Next, we study the S-contraction of Graph FRNNs.
\begin{thm}[Contraction of Graph FRNN]\label{thm:ignn_contraction}
Consider the dynamics~\eqref{eq:cts_graph_implicit_dynamics_vectorized} with a MONE nonlinearity $\Psi$. If
    \begin{enumerate}[label=\textup{(A\arabic*)},leftmargin=*]
        \item~\label{asP:adj_symmetry} The adjacency matrix satisfies $A = A^\top$.
        \item~\label{asP:adj_spectrum} The eigenvalues of $A$ lie in $[0,1]$.
        \item~\label{asP:graph_W_condn} $W$ satisfies matrix inequality~\eqref{eq:fr_cts_mone} for $P \succ 0$, diagonal $Q \succ 0$, rate $c > 0$, which satisfy $PQ^{-1}P \preceq 4(1-c)P$.
    \end{enumerate}
    Then, the dynamics~\eqref{eq:cts_graph_implicit_dynamics_vectorized} are strongly infinitesimally contracting with rate $c$ in the norm $\norm{\cdot}{I_n \kron P}$.
\end{thm}
\begin{proof}
    We show that the synaptic matrix $A^\top \kron W$ from~\eqref{eq:cts_graph_implicit_dynamics_vectorized} satisfies~\eqref{eq:fr_cts_mone} with $P$ and $Q$ replaced by $I_n \kron P$ and $I_n \kron Q$. Substituting these matrices into~\eqref{eq:fr_cts_mone} and applying the mixed product property (see Lemma~\ref{lem:mixed_product} in Appendix) yields
    \begin{align} \label{eq:graph_step_1}
        \begin{bmatrix}
        -2(1 - c)I_n \kron P & I_n \kron P + A \kron W^\top Q \\
        I_n \kron P + A \kron QW & -2 I_n \kron Q
        \end{bmatrix} \preceq 0,
    \end{align}
    where we have used $A = A^\top$. Since $A$ is symmetric, it admits the eigendecomposition $A = U \Lambda U^\top$, where each entry of $\Lambda$ (being an eigenvalue of $A$) lies in $[0,1]$. Applying the unitary transform $T = \diag{U \kron I_n, U \kron I_n}$ to~\eqref{eq:graph_step_1} and utilizing $U^\top U = UU^\top = I_n$ yields
    \begin{align}\label{eq:graph_step_2}
        \begin{bmatrix}
        -2(1 - c)I_n \kron P & * \\
        I_n \kron P + \Lambda \kron QW & -2I_n \kron Q
        \end{bmatrix} \preceq 0.
    \end{align}
    \textcolor{black}{It can be shown (e.g., by applying the Schur complement and examining the resulting block-diagonal matrix) that, since $\Lambda$ is diagonal, the condition~\eqref{eq:graph_step_2} splits into $n$ independent LMIs}
    \begin{align}
        \begin{bmatrix}
        -2(1 - c)P &  P + \lambda_i W^\top Q \\
         P + \lambda_i QW & -2Q
        \end{bmatrix} \preceq 0\;\forall i=1,\ldots,n,
    \end{align}
    where $\lambda_i$ is the $i$th diagonal element of $\Lambda$. Since this condition is linear in $\lambda_i$, and $\lambda_i \in [0,1]$, it suffices to verify it at $\lambda = 0$ and $\lambda = 1$. At $\lambda = 1$, this condition is identical to~\eqref{eq:fr_cts_mone}. At $\lambda = 0$, the condition is implied by~\ref{asP:graph_W_condn}, since
    \begin{align*}
        \begin{bmatrix}
        -2(1 - c)P &  P  \\
         P  & -2Q
        \end{bmatrix} \preceq 0 \iff PQ^{-1}P \preceq 4(1 - c)P.
    \end{align*}
    This finishes the proof.
\end{proof}

\begin{remark}[Computational Tractability for Large Graphs]
In applications such as implicit neural networks, only the fixed point of the FRNN dynamics is of interest, not the trajectory. In small-scale systems, alternative methods such as Peaceman-Rachford operator splitting can be used to find this fixed point, and accommodate a larger set of synaptic matrices than S-contraction requires. However, this approach scales poorly, particularly in the context of graph neural networks~\cite{JB-QW-CDH-BW:23}. Peaceman-Rachford requires to invert the matrix $(I_{nm} -  \alpha A^\top \kron W)$, for some constant $\alpha \in \R$. This operation is computationally expensive ($O(m^3n^3)$ in time, $O(m^2n^2)$ in memory), destroys the natural sparse structure of the adjacency matrix, and precludes decentralized execution. While the system~\eqref{eq:cts_graph_implicit_dynamics} restricts the set of synaptic matrices, it is computationally efficient
and permits a distributed implementation.
\end{remark}

\begin{remark}[Assumptions on Graph Structure]
We assume that the underlying graph is undirected. Restricting the eigenvalues to be in $[0,1]$ is common in Graph RNNs~\cite{TNK-MW:16}. Given an arbitrary adjacency matrix $\tilde{A}$, one can introduce the normalized adjacency matrix $A = \frac{1}{2}(\tilde D^{-1/2}\tilde{A}\tilde D^{-1/2} + I_n)$, where $\tilde D$ is the diagonal weight degree matrix for $\tilde{A}$. 
\end{remark}
Extending our contraction guarantees to directed graphs remains an open problem. 
{\color{black}
\begin{remark}[Comparisons to Small gain based results]
    General converse results for heterogeneous networks may be obtained via small gain theorems for contractivity~\cite[Theorem 3.23]{FB:26-CTDS} and exponential stability~\cite{CG-JDS-MAM:26}. Typically, these results treat other systems as disturbances, and as such, yield convergence rates worse than the rates of each subsystem, which we are able to maintain in the homogeneous S-contracting case.
\end{remark}
}

\begin{arxiv}
    An extension of our results to diagonally symmetrizable graphs is presented in Appendix~\ref{app:diagonal_symmetrizable_graphs}
\end{arxiv}

\section{Applications to Control design with RNNs}\label{sec:control_design}

In this section, we combine the separation principle (Theorem~\ref{thm:separation_principle}) with our S-contraction certificates (Theorem~\ref{thm:main_fr_h_conditions}) to design a reference tracking controller for a plant modeled by an RNN. We propose an LMI-based method for designing a full-state feedback controller that renders the plant contracting, alongside a contracting observer; similar contraction conditions for discrete-time models appear in~\cite{WDA-ALB-MF:24, AN-MSH-RDB:22}. We also establish necessary and sufficient feasibility conditions for the derived LMIs, drawing direct parallels to classical linear stabilizability and detectability. In accordance with Theorem~\ref{thm:separation_principle}, the interconnection of this observer and controller is exponentially stable. Finally, in order to achieve reference tracking, we design a low-gain integral controller motivated by~\cite{JWSP:21} and relying on Corollary~\ref{cor:separation_moving_target}.  
Throughout the section, the continuous-time FRNN model~\eqref{eq:cts_firing_rate} with MONE nonlinearities is considered; the analysis extends to all models from Table~\ref{tab:lmi_conditions}.

\subsection{Contraction via full state feedback}

We begin with an LMI-based design for a full-state feedback controller.

\begin{pro}[Contraction via state feedback controller]\label{prop:contracting_state_feedback}    Consider a plant described by an FRNN~\eqref{eq:cts_firing_rate} with a MONE nonlinearity, with access to the full state, i.e. $C = I_n$. For the control law $u = Kx$, for $K \in \R^{m \times n}$, the following hold:
    \begin{enumerate}
        \item The closed-loop system is an FRNN with state $x$, and the closed-loop weight matrix is $W_{cl} = W + BK$.
        \item The closed loop system is S-contracting for a given contraction rate $c \in (0,1]$ if and only if there exist a positive definite $X \in \R^{n\times n}$, a positive diagonal matrix $D$, and a matrix $Y \in \R^{m \times n}$ satisfying the LMI:
        \begin{equation}
            \label{eq:synthesis_lmi_sf}
            \begin{bmatrix}
            -2(1 - c)X & D + XW^\top + Y^\top B^\top \\
            D + WX + BY & -2D
            \end{bmatrix} \preceq 0.
        \end{equation}
        If~\eqref{eq:synthesis_lmi_sf} holds, the state feedback gain guaranteeing contraction is given by $K = YX^{-1}$.
    \end{enumerate} 
\end{pro}
\begin{proof}
    The proof is presented in Appendix~\ref{app:proof_feedback}
\end{proof}

\begin{cor}[Necessity of linear stabilizability]
    An FRNN with synaptic matrix $W$, and input matrix $B$ can be rendered S-contracting via full-state feedback with some $c > 0$ only if $(W - I_n, B)$ is stabilizable.
\end{cor}

\begin{proof}
    Pre- and post-multiplying~\eqref{eq:synthesis_lmi_sf} by $\begin{bmatrix} I_n & I_n \end{bmatrix}$ and its transpose, respectively, and simplifying yields
    \begin{equation*}
        WX + XW^\top - 2X + BY + Y^\top B^\top + 2cX \preceq 0.
    \end{equation*}
    Using the fact that $c> 0$, and $Y = KX$, we obtain
    \begin{equation*}
        (W - I_n + BK)X + X(W - I_n + BK)^\top \prec 0.
    \end{equation*}
    which is the condition for stabilizability of $(W - I_n, B)$.
\end{proof}

Similar to the problem of control design for linear systems, we may also apply the projection lemma~\cite[Section 2.6.2]{SB-LEG-EF-VB:94}.
\begin{cor}[Necessary and sufficient conditions for S-contraction via full state feedback]\label{cor:full_state_feedback_conditions}
    Let $\Pi_B$ denote a matrix whose columns form a basis for the null space of $B^\top$, such that $B^\top \Pi_B = 0$. An FRNN with synaptic matrix $W$, and input matrix $B$ can be rendered S-contracting via full state feedback with some $c > 0$ if and only if the following inequality holds
    \begin{equation}
\label{eq:projected_lmi_congruence}
        \begin{bmatrix} I_n & 0 \\ 0 & \Pi_B^\top \end{bmatrix}
        \begin{bmatrix}
        -2X & D + XW^\top \\
        D + WX & -2D
        \end{bmatrix}
        \begin{bmatrix} I_n & 0 \\ 0 & \Pi_B \end{bmatrix} \prec 0.
    \end{equation}
\end{cor}

\begin{proof}
    Let $U = 
    \begin{bmatrix} \0_{m \times n} & B^\top \end{bmatrix}$. Let $V = \begin{bmatrix} I_n & \0_{n \times n} \end{bmatrix}$.
    By continuity, the LMI~\eqref{eq:synthesis_lmi_sf} holds for some $ c>0 $ if and only if 
    \begin{align}\label{eq:step_1_controllability}
    \underbrace{\begin{bmatrix}
        -2X & D + XW^\top \\
        D + WX & -2D
        \end{bmatrix}}_{\mcM}
        + U^\top Y V
        + V^\top Y^\top U \prec 0.
    \end{align}
    An application of the Projection Lemma (Lemma~\ref{lem:projection_lemma}) states that the inequality~\eqref{eq:step_1_controllability} holds if and only if $U_{\perp}^\top \mcM U_\perp \prec 0$ and $V_{\perp}^\top \mcM V_\perp \prec 0$. The second condition always holds, as the (2,2) block of $\mcM$ is negative definite. The first condition is identical to inequality~\eqref{eq:projected_lmi_congruence}.
\end{proof}

\subsection{Contracting observer design}
We now consider the dual problem of state estimation.  We seek to estimate $x \in \R^n$, given the output $y \in \R^p$, and we propose the following observer architecture:
\begin{align}\label{eq:observer_dynamics}
    \dot{\hat{x}} &= -\hat{x} + \Psi(W\hat{x} + Bu + L(y - \hat{y})),\quad
    \hat{y} = C\hat{x},
\end{align}
where $\hat{x} \in \R^n$ is the estimated state, and $L$ is the observer gain matrix. 

\begin{pro}[S-Contracting observer design]\label{prop:contracting_observer}    Consider the observer~\eqref{eq:observer_dynamics} with exogenous inputs $u$ and $y$.
\begin{enumerate}
    \item The observer is an FRNN with state $\hat{x}$, and a weight matrix $W_{obs} = W - LC$.
    \item The observer is S-contracting for a given contraction rate $c \in (0,1]$ if and only if there exist $P \succ 0$, diagonal $Q \succ 0$ and a matrix $M$ such that, 
        \begin{equation}
            \label{eq:synthesis_lmi_obs}
            \begin{bmatrix}
            -2(1 - c)P & P + W^\top Q - C^\top M^\top \\
            P + QW - MC & -2Q
            \end{bmatrix} \preceq 0.
        \end{equation}       If~\eqref{eq:synthesis_lmi_obs} holds, the observer gain guaranteeing contraction is given by $L = Q^{-1}M$.
    \end{enumerate} 
\end{pro}
\begin{proof}
    The proof is presented in Appendix~\ref{app:proof_observer}
\end{proof}

\begin{cor}[Necessity of linear detectability]
An FRNN with synaptic matrix $W$, and output matrix $C$ can be observed by an S-contracting FRNN only if the pair $(W - I_n, C)$ is detectable.
\end{cor}
\begin{proof}
    Pre- and post-multiplying~\eqref{eq:synthesis_lmi_obs} by $\begin{bmatrix} I_n & I_n \end{bmatrix}$ and its transpose, respectively, yields
    \begin{equation*}
        QW + W^\top Q - 2Q - MC - C^\top M^\top + 2cP \preceq 0.
    \end{equation*}
    Substituting $M = QL$ and using the fact that $-2cP \prec 0$, 
    \begin{equation*}
        Q(W - I_n - LC) + (W - I_n - LC)^\top Q \prec 0 
    \end{equation*}
    which is the standard continuous-time Lyapunov condition for the pair $(W - I_n, C)$ to be detectable.
\end{proof}

\begin{cor}[Necessary and sufficient conditions for S-contracting observer design]
    Let $\Pi_C$ denote a matrix whose columns form a basis for the null space of $C$, such that $C \Pi_C = 0$.  An FRNN with synaptic matrix $W$, and output matrix $C$ can be observed by an S-contracting FRNN if and only if the following strict inequality holds for some $c > 0$:
    \begin{equation}
        \label{eq:projected_lmi_obs}
        \begin{bmatrix} \Pi_C^\top & 0 \\ 0 & I_n \end{bmatrix}
        \begin{bmatrix}
        -2P & P + W^\top Q \\
        P + QW & -2Q
        \end{bmatrix}
        \begin{bmatrix} \Pi_C & 0 \\ 0 & I_n \end{bmatrix} \prec 0.
    \end{equation}
\end{cor}

\begin{proof}
    Let $U = \begin{bmatrix} C & \0_{p \times n} \end{bmatrix}$. Let $V = \begin{bmatrix} \0_{n \times n} & I_n \end{bmatrix}$. Let $Y = -M^\top$.
    By continuity, the LMI~\eqref{eq:synthesis_lmi_obs} holds strictly for some $c>0$ if and only if 
    \begin{align}\label{eq:step_1_observability}
    \underbrace{\begin{bmatrix}
        -2P & P + W^\top Q \\
        P + QW & -2Q
        \end{bmatrix}}_{\mcM}
        + U^\top Y V
        + V^\top Y^\top U \prec 0.
    \end{align}
    An application of the Projection Lemma~\ref{lem:projection_lemma} states that the inequality~\eqref{eq:step_1_observability} holds if and only if $U_{\perp}^\top \mcM U_\perp \prec 0$ and $V_{\perp}^\top \mcM V_\perp \prec 0$. The second condition always holds, as the (1,1) block of $\mcM$ is negative definite. The first condition is identical to inequality~\eqref{eq:projected_lmi_obs} since $U_\perp = \begin{bmatrix} \Pi_C & 0 \\ 0 & I_n \end{bmatrix}$.
\end{proof}

\subsection{Low gain Integral control}

Having designed a controller in Proposition~\ref{prop:contracting_state_feedback}, and an observer in Proposition~\ref{prop:contracting_observer}, which satisfies the assumptions in our separation principle in Theorem~\ref{thm:separation_principle}, we now proceed to perform reference tracking via an integral controller and a singular perturbation argument. As a first step, we identify sufficient conditions for the contraction of the reduced dynamics.

\begin{lemma}[Contractivity of the reduced dynamics]\label{lem:dc_gain}
Let a map $\map{\xstar}{\R^m}{\R^n}$ exist that satisfies the equation
    \begin{align}
        \xstar(u) = \Psi(W\xstar(u) + Bu)
    \end{align}
    where $W \in \R^{n \times n}, B \in \R^{n \times m}$. Let $A = I_n - W, C \in \R^{p \times n}$, and $Q\in \R^{n\times n}$ be a diagonal positive matrix. Assume:
\begin{enumerate}[label=\textup{(A\arabic*)},leftmargin=*]
    \item $\Psi$ is slope-restricted in $[\delta, 1]$, for $\delta>0$.
    \item  There exists $\ustar$ such that $r = C\xstar(\ustar)$.
    \item With the shorthands ${Z = B^\top Q ((1 - \delta)I_n  + 2\delta A)}$ and ${R = 2\delta A^\top Q A + (1 - \delta) (QA + A^\top Q)}$, there exist matrices $Y$ and $P\succ 0$ such that
    \begin{align}\label{eq:dc_gain_contraction_lmi}
     \begin{bmatrix}
        2c_rP - 2 \delta B^\top QB &  Z -YC \\    * & -R
      \end{bmatrix}  \preceq 0.
    \end{align}
\end{enumerate}
Then, with $K = \varepsilon^{-1} P^{-1}Y$, the dynamical system
\begin{align}\label{eq:reduced_dynamics}
    \dot{u} = \varepsilon K(r - C\xstar(u))
\end{align}
is strongly infinitesimally contracting with rate $c_r$ in $\norm{\cdot}{P}$.
\end{lemma}
\begin{proof}
    The dynamics~\eqref{eq:reduced_dynamics} form a Lur'e system, where the nonlinearity is given by $\xstar(u)$ and admits an incremental multiplier (Lemma~\ref{lem:implicit_imm}). Applying Lemma~\ref{lem:absolute_contractivity_lure} to~\eqref{eq:reduced_dynamics}, one obtains the following LMI condition for contractivity in the norm $\norm{\cdot}{P}$: 
    \begin{align}\label{eq:int_condition_step_1}
     \begin{bmatrix}
        2c_rP & -\varepsilon PKC \\    -\varepsilon C^\top K^\top P & \0_{m\times{m}}
      \end{bmatrix} + 
       \begin{bmatrix}
            B & W \\ 0 & I_n 
        \end{bmatrix}^\top
        M 
        \begin{bmatrix}
            B & W \\ 0 & I_n 
        \end{bmatrix} \preceq 0.
    \end{align}
    where $M$ is the matrix multiplier for the nonlinearity $\Psi$. Using Lemma~\ref{lem:IMM_conditions}, we may substitute $M$ as $\begin{bmatrix}
            -2\delta Q &  (1 + \delta) Q \\
             (1 + \delta) Q & -2Q
        \end{bmatrix}$. The second term in~\eqref{eq:int_condition_step_1} resolves to
    \begin{align*}
        \begin{bmatrix}
         - 2 \delta B^\top QB &  - 2\delta B^\top QW + (1 + \delta) B^\top Q  \\    * &  (1 + \delta) (W^\top Q + QW) - 2Q - 2\delta W^\top QW 
      \end{bmatrix}
    \end{align*}
    Now, let us set $A = I_n - W$. First, note that 
    \begin{align*}
        - R &=- 2\delta W^\top QW + (1 + \delta) (W^\top Q + QW) - 2Q,  \\
        Z &= - 2\delta B^\top QW + (1 + \delta) B^\top Q.
    \end{align*}
    Now, we may rewrite~\eqref{eq:int_condition_step_1} using these simplifications. 
    \begin{align}\label{eq:dc_gain_contraction_lmi_prep}
     \begin{bmatrix}
        2c_rP - 2 \delta B^\top QB &  Z -\varepsilon PKC \\    * & -R
      \end{bmatrix}  \preceq 0.
    \end{align}
    Reparameterizing $Y = \varepsilon PK$ we get the condition~\eqref{eq:dc_gain_contraction_lmi}. 
\end{proof}

Finally, we may now design a low gain integral controller.

\begin{thm}[Reference Tracking]\label{thm:ref_tracking}
  Consider the FRNN plant~\eqref{eq:cts_firing_rate} with synaptic matrix
  $W \in \R^{n \times n}$, input matrix $B \in \R^{n \times m}$, output
  matrix $C \in \R^{p \times n}$, and reference signal $r \in \R^p$.
  Suppose the following conditions hold.
  \begin{enumerate}[label=\textup{(A\arabic*)},leftmargin=*]
    \item \label{asP:PlantC_FRNN} \textup{(Plant contraction by state-feedback)}
      There exists a matrix $K_f \in \R^{m \times n}$ such that the FRNN
      with synaptic matrix $W + BK_f$ is strongly infinitesimally
      contracting with rate $c_K > 0$ in the norm
      $\norm{\cdot}{\mcX}$.

    \item \label{asP:ObserverC_FRNN} \textup{(Observer contraction)}
      There exists a matrix $L \in \R^{n \times p}$ such that the FRNN
      with synaptic matrix $W - LC$ is strongly infinitesimally
      contracting with rate $c_O > 0$ in the norm
      $\norm{\cdot}{\mcO}$.

    \item \label{asP:ReducedC_FRNN} \textup{(Contraction of reduced-order dynamics)}
      For each constant input $u$, let $\xstar(u)$ denote the unique
      equilibrium of the FRNN with synaptic matrix $W + BK_f$ and input
      matrix $B$. Assume there exists $\ustar \in \R^m$ such that
      $r = C\xstar(\ustar)$. The reduced dynamics
      \begin{align}\label{eq:reduced_dynamics_FRNN}
        \dot{u}_{\mathrm{ext}}
          = \varepsilon\, K_i\bigl(r - C\xstar(u_{\mathrm{ext}})\bigr)
      \end{align}
      be strongly infinitesimally contracting with rate $c_r > 0$ in the
      norm $\norm{\cdot}{\mcR}$.

    \item \label{asP:epsilon_scale} \textup{(Low-gain condition)}
      Define the induced-norm constants
      \begin{alignat*}{2}
        \ell_u &= \norm{B}{\mcU \to \mcX},&\qquad
        \ell_K &= \norm{K_f}{\mcO \to \mcU},\\
        \ell_{i,\mcR} &= \norm{K_i C}{\mcX \to \mcR},&\qquad
        \ell_{i,\mcU} &= \norm{K_i C}{\mcX \to \mcU}.
      \end{alignat*}
      The gain parameter $\varepsilon > 0$ satisfies
      \begin{align}
        \varepsilon\bigl(\ell_{i,\mcU}\,\ell_u - c_K\, c_r\bigr)
          &< c_K^2, \label{eq:eps_cond1}\\
        \varepsilon\,\ell_{i,\mcU}\,\ell_u
          \bigl(c_K\, c_r + \ell_u\, \ell_{i,\mcR}\bigr)
          &< c_r\, c_K^3. \label{eq:eps_cond2}
      \end{align}
  \end{enumerate}

  \noindent
  Consider the closed-loop system
  \begin{subequations}\label{eq:closed_loop_FRNN}
  \begin{align}
    \dot{x}  &= -x + \Psi\bigl(Wx + BK_f\xi + B\, u_{\mathrm{ext}}\bigr),\; y = Cx
      \label{eq:FRNN_plant_controller}\\
    \dot{\xi} &= -\xi + \Psi\bigl(W\xi + BK_f\xi + L(y - C \xi)\bigr),
      \label{eq:FRNN_observer}
  \end{align}
  \end{subequations}
  driven by the integral controller
  \begin{align}\label{eq:integral_controller}
    \dot{u}_{\mathrm{ext}}
      = \varepsilon\, K_i\bigl(r - y\bigr).
  \end{align}
  Then the following statements hold.
  \begin{enumerate}[label=\textup{(\roman*)},leftmargin=*]
    \item \label{C:exponential_stability}
      \textup{(Global exponential stability)}
      For every constant $u_{\mathrm{ext}}$, the
      subsystem~\eqref{eq:closed_loop_FRNN} possesses a unique
      equilibrium $\bigl(\xstar(u_{\mathrm{ext}}),\,
      \xstar(u_{\mathrm{ext}})\bigr)$ that is globally exponentially
      stable, uniformly in $u_{\mathrm{ext}}$.

    \item \label{C:reference_tracking}
      \textup{(Reference tracking)}
      The full closed-loop
      system~\eqref{eq:closed_loop_FRNN}--\eqref{eq:integral_controller}
      possesses the globally exponentially stable equilibrium
      $\bigl(\xstar(\ustar),\, \xstar(\ustar),\, \ustar\bigr)$.
      In particular, $\lim_{t\to\infty} Cx(t) = r$.
  \end{enumerate}
\end{thm}

\begin{proof}
By Theorem~\ref{thm:separation_principle}, Assumptions~\ref{asP:PlantC_FRNN}-\ref{asP:ObserverC_FRNN} and the Lipschitz nature of the controller ensure the closed-loop system~\eqref{eq:closed_loop_FRNN} is globally exponentially stable to the unique fixed point $(x^\star(\uext), x^\star(\uext))$, uniformly in $\uext$. 

Now, consider the functions, $V_1 = \norm{\xi - x}{\mcO}$, $V_2 = \norm{x - \xstar(\uext)}{\mcX}$, and $V_3 = \norm{\uext - \ustar}{\mcR}$. We will bound each of their Dini derivatives. First, from Assumption~\ref{asP:ObserverC_FRNN}
\begin{align}\label{eq:V1_dot}
    D^+ V_1 \leq -c_O V_1.
\end{align}
Next, utilizing Corollary~\ref{cor:separation_moving_target} for $V_2$ and treating the low gain integral controller as the time varying parameter, 
\begin{align}\label{eq:V2_bound_1}
D^+ V_2 &\leq -c_K V_2 + \ell_u \ell_K \norm{\xi - x}{\mcO} + \frac{\ell_u}{c_K} \norm{\dot{u}_{\text{ext}}}{\mcU} 
\end{align}

Substituting $r = C\xstar(\ustar)$ in the expression for $\norm{\dot{u}_{\text{ext}}}{\mcU}$, 
\begin{align*}
\norm{\dot{u}_{\text{ext}}}{\mcU} &= \norm{\varepsilon K_i C(\xstar(\ustar) - x)}{\mcU}\\&= \norm{\varepsilon K_i C(\xstar(\ustar) - \xstar(\uext) + \xstar(\uext) - x)}{\mcU}\\
&\leq \varepsilon  \ell_{i,\mcU}(\norm{\xstar(\ustar) - \xstar(\uext))}{\mcX} + V_2)
\end{align*}

Since $\xstar(u)$ is globally Lipschitz in $u$ with constant $\ell_u/ c_K$, 
\begin{align*}
\norm{\dot{u}_{\text{ext}}}{\mcU} 
&\leq \frac{\varepsilon \ell_{i,\mcU} \ell_u}{c_K}V_3 + \varepsilon \ell_{i,\mcU}V_2
\end{align*}

Substituting $\norm{\dot{u}_{\text{ext}}}{\mcU}$ back into~\eqref{eq:V2_bound_1} yields:
\begin{align} \label{eq:V2_dot}
D^+ V_2 \leq -(c_K - \varepsilon \frac{\ell_{i,\mcU}\ell_u}{c_K}) V_2 + \ell_u \ell_K V_1 + \varepsilon \frac{ \ell_{i,\mcU} \ell_u^2}{c_K^2} V_3,
\end{align}
In order to find the Dini derivative of $V_3$, we first rewrite the integral controller~\eqref{eq:integral_controller} as
\begin{align*}
\dot{u}_{\text{ext}} 
&= \varepsilon K_iC(\xstar(\ustar) - C\xstar(\uext)) + \varepsilon K_i C(\xstar(\uext) - x).
\end{align*}
Treating the second term as a disturbance to the contracting reduced order dynamics~\eqref{eq:reduced_dynamics_FRNN}, from the incremental ISS property of contracting systems~\cite[Theorem 3.16]{FB:26-CTDS},
\begin{align}
D^+ V_3 &\leq -\varepsilon c_r V_3 + \varepsilon \ell_{i, \mcR}\norm{\xstar(\uext) - x}{\mcX} \nonumber \\
&= -\varepsilon c_r V_3 + \varepsilon \ell_{i,\mcR} V_2. \label{eq:V3_dot}
\end{align}

Upon stacking inequalities~\eqref{eq:V1_dot},~\eqref{eq:V2_dot},~\eqref{eq:V3_dot}, we obtain the following inequality,
\begin{align*}
    D^+\begin{bmatrix}
        V_1 \\
        V_2 \\
        V_3
    \end{bmatrix} 
    \leq 
\underbrace{
    \begin{bmatrix}
        -c_O & 0 & 0 \\
        \ell_u \ell_K &  -(c_K - \varepsilon \frac{\ell_{i,\mcU}\ell_u}{c_K}) & \varepsilon \frac{\ell_{i,\mcU}\ell_u^2}{c_K^2} \\
        0 & \varepsilon \ell_{i,\mcR} & - \varepsilon c_r
    \end{bmatrix}
}_{M(\varepsilon)}
    \begin{bmatrix}
        V_1 \\
        V_2 \\
        V_3
    \end{bmatrix}
\end{align*}

Next, we show that $M(\varepsilon)$ is Hurwitz. The eigenvalues of $M(\varepsilon)$ are obtained upon solving $\det(M(\varepsilon) - \lambda I_3) = 0$. One eigenvalue is $-c_O$, and the other two must be determined by the lower right $2\times 2$ block. For both of the remaining eigenvalues to be negative, this block must have a negative trace, and a positive determinant. Both these conditions are met under assumption~\ref{asP:epsilon_scale}.

By continuity,  $M(\varepsilon) + \eta I_3$ remains Hurwitz for small $\eta> 0$. Since $M(\varepsilon) + \eta I_3$ is Metzler, this is equivalent to guaranteeing the existence of an $\alpha = [\alpha_1, \alpha_2, \alpha_3]^\top$, such that $(M(\varepsilon) + \eta I_3)^\top \alpha$ is elementwise negative. Therefore, denoting $V = \alpha_1 V_1 + \alpha_2 V_2 + \alpha_3 V_3$, one has
    $D^+ V \leq -\alpha^\top M(\varepsilon) V \leq -\eta V$,
guaranteeing exponential stability to the fixed point $(\xstar(\ustar), \xstar(\ustar), \ustar)$. Since $r = C\xstar(\ustar)$, the output of the plant settles exponentially to the value $r$.
\end{proof}

\subsection{Numerical Validation}
To validate our approach numerically, we apply our controller synthesis on a normalized version of the two tank benchmark system~\cite{MS-JPN:17}. We performed system identification using Neuromancer~\cite{JD-AT-JK-MS-BJ-DV:23}, fitting the plant to an FRNN with $n = 8$ neurons and $\tanh{(\cdot)}$ activations. We apply Propositions~\ref{prop:contracting_state_feedback} and~\ref{prop:contracting_observer} to design a controller and an observer respectively, and design a gain for the integral controller via Lemma~\ref{lem:dc_gain}. As shown in Figure~\ref{fig:data_driven_control}, our control architecture successfully tracks piecewise constant references. 
\begin{figure}
    \centering
    \includegraphics[width=0.97\linewidth]{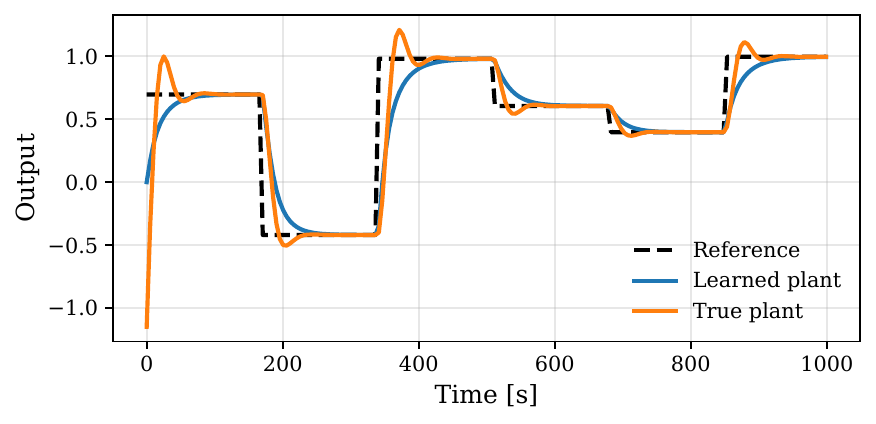}
    \caption{ For a two tank system modeled by an FRNN, we utilize our design mechanism for the full state feedback controller, the contracting observer and the integral gain to design a closed loop system capable of tracking references, validating the proposed theoretical results.}
    \label{fig:data_driven_control}
\end{figure}

\section{Applications in Machine Learning}\label{sec:machine_learning}
We now utilize S-contraction for Deep Equilibrium Models (DEQs)~\cite{SB-JZK-VK:19}. DEQs replace finite-depth architectures by defining their output as solution of an implicit equation, which in turn is solved by iterating a dynamical system. For these models to be well-posed, the equilibrium must be unique and globally asymptotically stable—properties natively guaranteed by a contracting continuous-time FRNN.

We first derive an unconstrained parameterization of weight matrices that satisfies our contractivity LMI by construction. We then utilize this parameterization to design an architecture where the weight and input matrices are input dependent. This allows the DEQ to model locally Lipschitz functions and improves parameter efficiency.

\subsection{Parameterization}
We first derive an unconstrained parameterization for synaptic matrices of S-contracting FRNNs.
\begin{thm}[Parameterization of weight matrices]\label{thm:parameterization}
    The following  statements are equivalent:
    \begin{enumerate}
        \item The synaptic matrix $W$ satisfies the continuous time MONE FRNN certificate~\eqref{eq:fr_cts_mone} for some $P \succ 0$, diagonal matrix $Q\succ 0$, and $c \in [0,1]$.
        \item $W$ can be written as
            \begin{align*}          
                W = 2\sqrt{1 {-} c} \; \diag{\e^d} S V^\top V  - \diag{\e^{2d}}(V^\top V )^2 
            \end{align*}
            where $S \in \R^{n\times n}$ satisfies $S^\top S \preceq I_n$, $V$ is a full rank matrix in $\R^{n\times n}$, and $d\in \R^n$. 
    \end{enumerate}
\end{thm}
\begin{proof}
    Using the Schur complement, the inequality~\eqref{eq:fr_cts_mone} is equivalent to 
    \begin{align}\label{eq:begin_param}
        (P + W^\top Q)Q^{-1} (P + QW) \preceq 4(1-c)P.
    \end{align}
    From the Douglas-Fillmore-Williams Lemma~\cite[Theorem 8.6.2]{DSB:09}, the inequality~\eqref{eq:begin_param} is true if and only if there exists a matrix $S \in \R^{n\times n}$ satisfying $S^\top S \preceq I_n$ and
    \begin{align}
        Q^{-1/2} (P + QW) = 2\sqrt{(1-c)}\;SP^{1/2}.
    \end{align}
    Upon rearranging to solve for $W$, 
     \begin{align}
        W = 2\sqrt{1 {-} c} \; Q^{-1/2} S P^{1/2} - Q^{-1}P.
    \end{align}
    Now, the set of all valid $Q$ can be parameterized using $\diag{\e^{-2d}}$, and the
    set of all $P \succ 0$ can be parameterized as $V^\top V$ for full rank $V \in \R^{n\times n}$.
\end{proof}
\begin{remark}[Unconstrained Optimization Formulation]
In numerical implementations, the set of all matrices $S$ such that $S^\top S \preceq I_n$ may be parameterized via a free variable $X \in \R^{n\times n}$ by setting $S=X(I+X^{\top}X)^{-1/2}$. Similarly, to ensure $V^\top V$ does not lose rank, we may approximate $V^\top V$ using a free matrix $Y \in \R^{n \times n}$ and a small constant $\epsilon > 0$ by setting $P = Y^\top Y + \epsilon I_n$.
\end{remark}
Theorem~\ref{thm:parameterization} provides a constructive mechanism to design continuous-time FRNNs that are contracting by construction. By optimizing over the free variables $\{X, Y, d\}$, stability is preserved during learning.

\subsection{Applications to implicit neural networks} \label{sec:INN}
Recent literature~\cite{JL-LD-SO-WY:25} suggests that increasing the number of iterations of a DEQ improves its expressivity. This is due to the fact that one can model locally Lipschitz maps via fixed point equations. In order to canonically integrate this insight into a DEQ, we propose the following architecture.

\begin{pro}[Local Lipschitz continuity of the equilibrium map]
Let $\map{W}{\mcU}{\R^{n \times n}}$ and $\map{B}{\mcU}{\R^{n}}$ be globally Lipschitz with constants $\ell_W$ and $\ell_B$. Let $\Psi:\mathbb{R}^n\to\mathbb{R}^n$ be MONE. Define the map $\map{\xstar}{\mcU}{\R^n}$ by the implicit equation
\begin{align}
        \xstar(u) = \Psi(W(u)\xstar(u) + B(u)).
\end{align}
If $W(u) - I_n$ is Lyapunov diagonally stable uniformly in $u$, then the map $\xstar(u)$ is locally Lipschitz for each $u \in \mcU$.
\end{pro}

\begin{proof}
Let $u, u' \in \mcU$, and let $x := \xstar(u), x' := \xstar(u')$. Set
$\varepsilon := (W(u){-}W(u'))x' + B(u){-}B(u')$. 
Since $W - I_n$ is LDS, there must exist a positive diagonal $Q \succ 0$ and $\delta > 0$ such that $QW(u) + W(u)^\top Q \leq 2Q - 2\delta I_n$. Using the MONE incremental matrix multiplier, 
\begin{align}
    \norm{x - x'}{Q}^2 \leq (x{-}x')^\top Q(W(u)(x{-}x') + \varepsilon).
\end{align}
Applying the LDS bound to the first term on the right cancels $\norm{x - x'}{Q}^2$
from both sides, yielding,
\begin{align}
    \delta\norm{x - x'}{2}^2 \leq (x{-}x')^\top Q\,\varepsilon.
\end{align}
Through Cauchy Schwarz and global Lipschitz continuity of $W, B$ we obtain the Lipschitz bound,
\begin{align}
    \norm{\xstar(u) {-} \xstar(u')}{} \leq 
    \frac{\norm{D}{}}{\delta}\bigl(\ell_W\|\xstar(u')\| + \ell_B\bigr)\|u {-} u'\|.\label{eq:key}
\end{align}

Fixing $u_0 \in \mathcal{U}$ and $r>0$, applying \eqref{eq:key} with $u'=u_0$ yields a
uniform bound $\|\xstar(u')\| \leq R(u_0,r) < \infty$ for all $u' \in \mathcal{B}(u_0,r)$.
Substituting back into \eqref{eq:key} gives Lipschitz constant
$L(u_0,r) = \tfrac{\|D\|}{\delta}(\ell_W R(u_0,r)+\ell_B) < \infty$
on $\mathcal{B}(u_0,r)$.
\end{proof}

To make $W$ and $B$ input-dependent while maintaining the contractivity of the FRNN, we parameterize $W$ according to Theorem~\ref{thm:parameterization}, replacing each free variable and $B$ with an input-dependent feedforward neural network. Because these feedforward networks typically consist of compositions of linear transformations and MONE nonlinearities, they are globally Lipschitz. This allows us to design a highly expressive, locally Lipschitz implicit neural network that remains mathematically guaranteed to be well-posed and strictly contracting.  We validate this approach on the MNIST and CIFAR-10 image classification benchmarks. As shown in Table~\ref{tab:image_classification}, our model achieves competitive accuracy while remaining parameter efficient, due to a higher expressivity. Full implementation details and code are provided in our repository.\footnote{\url{https://github.com/AnandGokhale/Contractivity_Neural_Networks}}

\begin{table}[htbp]
    \centering
    \caption{Comparative Performance on Image Classification Benchmarks}
    \label{tab:image_classification}
    \begin{tabular}{lcc}
        \toprule
        \textbf{Method} & \textbf{Model size} & \textbf{Acc.} \\
        \midrule
        \multicolumn{3}{c}{\textbf{MNIST}} \\
        \midrule
        LBEN~\cite{MR-RW-IRM:20} & -- & 98.2\% \\
        monDEQ~\cite{EW-JZK:20}  & 84K & 99.1$\pm$0.1\% \\
        \textbf{Ours} & \textbf{89K} & \textbf{99.33\%} \\
        \midrule
        \multicolumn{3}{c}{\textbf{CIFAR-10}} \\
        \midrule
        LBEN~\cite{MR-RW-IRM:20} & -- & 71.6\% \\
        monDEQ~\cite{EW-JZK:20}  & 172K & 74.0$\pm$0.1\% \\
        monDEQ$^*$~\cite{EW-JZK:20}  & 854K & 82.0$\pm$0.3\% \\
        \textbf{Ours} & \textbf{134K} & \textbf{78.27\%} \\
        \textbf{Ours}$^*$ & \textbf{134K} & \textbf{82.30\%} \\
        \bottomrule
    \end{tabular}
    
    \vspace{2pt}
    \raggedright \footnotesize $^*$ indicates models trained with data augmentation.
\end{table}

\section{Conclusion}
In this paper, we established a contraction-based nonlinear separation principle and developed a comprehensive framework for the robust analysis, control design, and machine learning deployment of recurrent neural networks (RNNs). First, we introduced the separation principle alongside its parametric extensions for robustness and equilibrium tracking. Next, we derived sharp contraction certificates for FRNNs and HNNs, and investigated the fundamental properties of RNN interconnections, including an extension of our certificates to Graph RNNs. Building upon these theoretical foundations, we proposed a solution to the reference tracking problem for plants modeled by RNNs via LMI-based synthesis and low-gain integral control. Finally, we applied our certificates to implicit deep learning by deriving an exact, unconstrained algebraic parameterization. This parameterization enabled the design of parameter-efficient implicit neural networks that demonstrated improved expressivity and competitive benchmark accuracy.

Several promising directions for future work remain.
First, we utilize our parametric extension of the separation principle with a low gain controller, an 
interesting direction is to consider extensions to optimization based controllers~\cite{MC-ED-AB:20}. Second, extending the Graph RNN to undirected graphs would significantly broaden its applicability. Third, while our current analysis is restricted to standard RNNs, extending this contraction-based framework to more complex sequence-modeling architectures, such as LSTMs or transformers, is of major theoretical interest. Finally, deploying our robust contraction certificates in broader application domains, such as learning-to-optimize and large-scale network control, present an exciting avenue for future research.

\appendices
\section{Algebraic results}
The following result is described in~\cite[Section 2.6.2]{SB-LEG-EF-VB:94}
\begin{lemma}[Projection Lemma]\label{lem:projection_lemma}
Let $U \in \R^{m \times p}$ and $V \in \R^{n \times p}$, and let $Q = Q^\top \in \R^{p \times p}$. There exists a matrix $X \in \R^{m\times n}$ satisfying,
\begin{align}
    Q + U^\top XV + V^\top XU \prec 0 
\end{align}
if and only if
\begin{align}
    U_{\perp}^\top QU_{\perp} \prec 0 \quad \text{and} \quad V_{\perp}^\top QV_{\perp} \prec 0.
\end{align}
    
\end{lemma}

Next, we introduce an incremental matrix multiplier for this fixed point map from~\cite{LDA-MC:13}.
\begin{lemma}\label{lem:implicit_imm}
    Let the map $\map{\Psi}{\R^n}{\R^n}$ admit an incremental matrix multiplier $M$. Let the map $\map{\xstar}{\R^m}{\R^n}$ be well-posed and defined as the solution of the implicit equation,
    \begin{align}
        \xstar(u) = \Psi(W\xstar(u) + Bu)
    \end{align}
    where $W \in \R^{n \times n}, B \in \R^{n \times m}$. Then the map $\xstar(u)$ admits the incremental  multiplier matrix
    \begin{align}
        \begin{bmatrix}
            B & W \\ 0 & I_n 
        \end{bmatrix}^\top
        M
        \begin{bmatrix}
            B & W \\ 0 & I_n 
        \end{bmatrix}.
    \end{align}
\end{lemma}

\begin{lemma}[Mixed Product Property]\label{lem:mixed_product}
    Consider matrices $A \in \R^{n \times m}, B \in \R^{m \times p}, C \in \R^{k \times l}, D \in \R^{l \times q}$. $AB \kron CD = (A \kron C)(B \kron D)$.
\end{lemma}

\section{Proof for Lemma~\ref{lem:absolute_contractivity_lure}}\label{app:lure_proof}
\begin{proof}
  Given ${x_1,x_2\in\R^n}$, adopt the shorthands ${\Delta x =
  x_1-x_2\in\real^n}$, ${\Delta y = y_1-y_2 = H\Delta x\in\real^m}$, and ${\Delta
  \Psi = \Psi(y_1)-\Psi(y_2)\in\real^m}$. Since
  \begin{equation*}
    \begin{bmatrix} \Delta y \\ \Delta \Psi \end{bmatrix}
    =
    \begin{bmatrix} H \Delta x  \\ \Delta \Psi \end{bmatrix}
    =
    \begin{bmatrix} H & \0_{m\times{m}} \\ \0_{m\times{n}} & I_m   \end{bmatrix} 
    \begin{bmatrix} \Delta x \\ \Delta \Psi \end{bmatrix}, 
  \end{equation*}
  we rewrite the incremental multiplier matrix condition~\eqref{eq:imm-d}
  as
  \begin{align} \label{eq:imm-x}
    \begin{bmatrix} \Delta y \\ \Delta \Psi \end{bmatrix}^\top
    M 
    \begin{bmatrix} \Delta y \\ \Delta \Psi \end{bmatrix}
    \geq 0
    \iff
    \begin{bmatrix} \Delta x \\ \Delta \Psi \end{bmatrix}^\top M_H
    \begin{bmatrix} \Delta x \\ \Delta \Psi \end{bmatrix}
    \geq 0.
  \end{align}
The continuous time result is a special case of the result in~\cite[Theorem 4.2]{LDA-MC:13}. In the discrete time case, right and left multiplying~\eqref{eq:disc_lure_lmi} by $\begin{bmatrix} \Delta x^\top & \Delta \Psi^\top \end{bmatrix}^\top$ and its transpose respectively, and using~\eqref{eq:imm-x} we get,
\begin{align}
         \begin{bmatrix} \Delta x \\ \Delta \Psi \end{bmatrix}^\top  
            \begin{bmatrix}
        A^\top P A - \rho^2 P &  A^\top P B  \\
        B^\top P A &  B^\top P B
      \end{bmatrix}
      \begin{bmatrix} \Delta x \\ \Delta \Psi \end{bmatrix}
      \leq 0 \\
      \iff \norm{A\Delta x + B \Delta \Psi}{P}^2 \leq \rho^2 \norm{\Delta x}{P}^2.
    \end{align}  
This is precisely the contraction condition for discrete-time systems.
\end{proof}

\section{Proofs for Section~\ref{sec:control_design}}
\subsection{Proof of Proposition~\ref{prop:contracting_state_feedback}} \label{app:proof_feedback}

\begin{proof}
    Under the full state feedback control law $u = Kx$, the closed-loop dynamics are given by
    \begin{align}
    \dot{x} &= -x + \Psi(Wx + BKx) = -x + \Psi(W_{cl}x),
    \end{align}
    where $W_{cl} = W + BK$. For a given contraction rate $c \in (0,1]$, the contraction condition~\eqref{eq:fr_cts_mone} for $W_{cl}$ is
    \begin{equation}
        \label{eq:proof_lmi_initial_sf}
        \begin{bmatrix}
        -2(1 - c)P & P + (QW_{cl})^\top \\
        P + QW_{cl} & -2Q
        \end{bmatrix} \preceq 0,
    \end{equation}
    where $P \succ 0$ is a symmetric matrix and $Q \succ 0$ is a diagonal matrix. To convert this expression into an LMI, we pre- and post-multiply~\eqref{eq:proof_lmi_initial_sf} by the block diagonal matrix $\diag{P^{-1}, Q^{-1}}$. This yields the equivalent inequality:
    \begin{align}
        \begin{bmatrix}
        -2(1 - c)P^{-1} & Q^{-1} + P^{-1}W_{cl}^\top \\
        Q^{-1} + W_{cl}P^{-1} & -2Q^{-1}
        \end{bmatrix} \preceq 0.
    \end{align}
    Applying the change of variables $X = P^{-1}$ and $D = Q^{-1}$, and substituting $W_{cl} = W + BK$, the $(2,1)$ block becomes
    \begin{align}
        D + W_{cl}X &= D + WX + BKX.
    \end{align}
    Defining the variable $Y = KX$, the transformed inequality is rendered into an LMI in the variables $(X, D, Y)$, given by, 
    \begin{equation}
        \begin{bmatrix}
        -2(1 - c)X & D + XW^\top + Y^\top B^\top \\
        D + WX + BY & -2D
        \end{bmatrix} \preceq 0.
    \end{equation}
    The original $P$ and $Q$ matrices may be recovered via ${P = X^{-1}}$ and $Q = D^{-1}$, and the stabilizing state feedback gain is uniquely recovered via $K = YX^{-1}$.
\end{proof}
\subsection{Proof of Proposition~\ref{prop:contracting_observer}}
\label{app:proof_observer}

\begin{proof}
    We begin by noting that substituting $\hat{y} = C\hat{x}$ into the observer dynamics~\eqref{eq:observer_dynamics} yields an FRNN with synaptic matrix $W_{obs} = W - LC$. Writing~\eqref{eq:fr_cts_mone} for $W_{obs}$, substituting $M = QL$, we obtain, 
    \begin{equation}
        \begin{bmatrix}
        -2(1 - c)P & P + W^\top Q - C^\top M^\top \\
        P + QW - MC & -2Q
        \end{bmatrix} \preceq 0.
    \end{equation}  
    This is an LMI jointly in $P, Q, M$, and the gain matrix $L$ may be recovered by computing $Q^{-1}M$.
\end{proof}
\begin{arxiv}
\section{Extension to diagonally symmetrizable graphs}
\label{app:diagonal_symmetrizable_graphs}
\begin{thm}[Contraction of Graph FRNN]\label{thm:ignn_contraction}
 Consider the dynamics~\eqref{eq:cts_graph_implicit_dynamics_vectorized} with a MONE nonlinearity $\Psi$.  If
    \begin{enumerate}[label=\textup{(A\arabic*)},leftmargin=*]
        \item~\label{asP:adj_symmetry} The adjacency matrix $A$ can be decomposed as $A = HD^{-1} $, where $D \succ 0$ is diagonal and $H = H^\top$.
        \item~\label{asP:adj_spectrum} The eigenvalues of $A$ lie in $[0,1]$.
        \item~\label{asP:graph_W_condn} $W$ satisfies matrix inequality~\eqref{eq:fr_cts_mone} for $P \succ 0$, diagonal $Q \succ 0$, rate $c > 0$, which satisfy  $PQ^{-1}P \prec 4(1 - c)P$. 
    \end{enumerate}
    Then, the dynamics~\eqref{eq:cts_graph_implicit_dynamics_vectorized} are strongly infinitesimally contracting with rate $c$ in the norm $\norm{\cdot}{D \kron P}$.
\end{thm}
\begin{proof}
    We show that the weight matrix $A^\top \kron W$ satisfies~\eqref{eq:fr_cts_mone} for weights $D \kron P$ and $D \kron Q$. Substituting these weights into~\eqref{eq:fr_cts_mone}, and using the mixed product property of Kronecker products Lemma~\ref{lem:mixed_product} yields
    \begin{align} \label{eq:graph_step_1}
        \begin{bmatrix}
        -2(1 - c)D \kron P & D \kron P + H \kron W^\top Q \\
        D \kron P + H^\top \kron QW & -2 D \kron Q
        \end{bmatrix} \preceq 0
    \end{align}
    Next, we apply a congruence transformation. Left and right multiplying~\eqref{eq:graph_step_1} by the non-singular block diagonal matrix $T_Q = \diag{D^{-1/2} \kron I_n, D^{-1/2} \kron I_n}$ yields,
    \begin{align} \label{eq:graph_step_2}
        \begin{bmatrix}
        -2(1 - c)I_n \kron P & *\\
        I_n \kron P + \bar{H}^\top \kron QW & -2 I_n\kron Q
        \end{bmatrix} \preceq 0,
    \end{align}
    where $\bar{H} = D^{-1/2}HD^{-1/2} $.
    Since $H$ is symmetric and $D$ is diagonal, $\bar{H}$ is also symmetric. Further, the eigenvalues of $A$ are identical to that of $\bar{H}$, as $D^{1/2}AD^{-1/2} = \bar{H}$. 

    Performing an SVD for $\bar{H}$ yields $\bar{H} = U \Lambda U^\top$, due to its symmetric nature. Further, each entry of $\Lambda$ is in $[0,1]$. We now apply the unitary transform, $T = \diag{U \kron I_n, U \kron I_n}$ to~\eqref{eq:graph_step_2}. Utilizing the fact that $U^\top U = UU^\top = I_m$,~\eqref{eq:graph_step_2} is equivalent to  
    \begin{align}\label{eq:graph_step_3}
        \begin{bmatrix}
        -2(1 - c)I_n \kron P & * \\
        I_n \kron P + \Lambda \kron QW & -2I_n \kron Q
        \end{bmatrix} \preceq 0.
    \end{align}
    Due to the diagonal nature of $I_n$ and $\Lambda$, this matrix may be permuted into a block diagonal form through the perfect shuffle permutation~\cite[Proposition 1]{DJR:80}. Therefore, the condition~\eqref{eq:graph_step_3} is decomposable into $n$ conditions, and the $i$th condition has the form,
    \begin{align}
        \begin{bmatrix}
        -2(1 - c)P &  P + \lambda_i W^\top Q \\
         P + \lambda_i QW & -2Q
        \end{bmatrix} \preceq 0,
    \end{align}
    where $\lambda_i$ is the $i$th diagonal element of the $\Lambda$. Since this condition is linear in $\lambda_i$, and $\lambda_i \in [0,1]$. this condition is true as long as it holds for $\lambda = 0$ and $\lambda = 1$. At $\lambda = 1$, this condition is identical to~\eqref{eq:fr_cts_mone}. At $\lambda = 0$, we obtain,
    \begin{align*}
        \begin{bmatrix}
        -2(1 - c)P &  P  \\
         P  & -2Q
        \end{bmatrix} \preceq 0 \iff PQ^{-1}P \prec 4(1 - c)P.
    \end{align*}
    The equivalence occurs due to the Schur complement along the $(2,2)$ block. Since $W$ satisfies~\eqref{eq:fr_cts_mone}, and $PQ^{-1}P \prec 4(1 - c)P$ as per~\ref{asP:graph_W_condn}, the proof holds.
\end{proof}

\begin{remark}[Assumptions on Graph Structure]
The decomposition $A = HD^{-1}$ restricts the network topology to diagonally symmetrizable graphs. This naturally encompasses undirected graphs when $D = I_n$. Furthermore, the condition on the eigenvalues~\ref{asP:adj_spectrum} 
can be performed by preprocessing the adjacency matrix $A = \frac{1}{2}(\frac{A_{orig}}{\norm{A_{orig}}{}} + I_n)$.
\end{remark}
\end{arxiv}

\bibliography{alias, Main, FB}

\begin{tac}
    
\begin{IEEEbiography}
	[{\includegraphics[width=1in,height=1.25in,clip,keepaspectratio]{coworkers/anand-gokhale.jpg}}]{Anand Gokhale} is a Ph.D. candidate in mechanical engineering at the University of California, Santa Barbara with the Mechanical Engineering Department and the Center for Control, Dynamical Systems, and Computation. He received the B.Tech and M.Tech degrees in Electrical Engineering from the Indian Institute of Technology Madras, India, in 2022. His research interests include contraction theory optimization, neural networks, human-robot teaming and agentic architectures.
\end{IEEEbiography}
\begin{IEEEbiography}[{\includegraphics[width=1in,height=1.25in,clip,keepaspectratio]{coworkers/anton-proskurnikov.png}}]{Anton V. Proskurnikov}(Senior Member, IEEE)
is an Associate Professor with the Department of Electronics and Telecommunications, Politecnico di Torino, Italy. He
received the M.Sc. and Ph.D. degrees in applied mathematics from St. Petersburg State University, St. Petersburg, Russia, in 2003 and 2005, respectively. 
He was previously with St. Petersburg State University (2003–2010), the Russian Academy of Sciences (2006–2022), the University of Groningen (2014–2016), and Delft University of Technology (2016–2018).  His research interests include complex
networks, nonlinear dynamics, and agent-based models in social
and natural sciences.
\end{IEEEbiography}

\begin{IEEEbiography}[{\includegraphics[width=1in,height=1.25in,clip,keepaspectratio]{coworkers/yu-kawano.jpg}}]{Yu Kawano} (M'13)  %[{\includegraphics[width=1in,height=1.25in,clip,keepaspectratio]{a1.jpg}}]
has been a Full Professor in the Graduate School of Advanced Science and Engineering at Hiroshima University since 2026. He received his M.S. and Ph.D. degrees in Engineering from Osaka University, Japan, in 2011 and 2013, respectively. He was a Postdoctoral Researcher at Kyoto University, Japan, from 2013 to 2016, and at the University of Groningen, the Netherlands, from 2016 to 2019. He was an Associate Professor at Hiroshima University, from 2019 to 2026. He has also held visiting research positions at Tallinn University of Technology (Estonia), the University of Groningen (the Netherlands), the University of Pavia (Italy), the Indian Institute of Technology Bombay (India), and the University of California, Santa Barbara (USA). His research interests include nonlinear systems, complex networks, model reduction, and privacy in control systems. He serves as an Associate Editor for Systems \& Control Letters, IEEE Transactions on Systems, Man, and Cybernetics: Systems, IEEE CSS Conference Editorial Board, and EUCA Conference Editorial Board.
\end{IEEEbiography}

\begin{IEEEbiography}
	[{\includegraphics[width=1in,height=1.25in,clip,keepaspectratio]{coworkers/francesco-bullo.jpg}}]{Francesco Bullo}
	(Fellow, IEEE) is a Distinguished Professor of Mechanical Engineering with
the University of California, Santa Barbara, CA, USA. He was previously
with the University of Padova (Laurea degree, 1994), Italy, the California
Institute of Technology (Ph.D. degree, 1998), Pasadena, CA, and the
University of Illinois at Urbana-Champaign, IL, USA. His research interests
include contraction theory, network systems, and distributed control.  He
is the author or coauthor of Geometric Control of Mechanical Systems
(Springer, 2004), Distributed Control of Robotic Networks (Princeton,
2009), Lectures on Network Systems (KDP, v1.7, 2024), and Contraction Theory for
Dynamical Systems (KDP, v1.2, 2024).  He served as IEEE CSS President and
SIAG CST Chair.  He is a Fellow of ASME, IFAC, and SIAM.

\end{IEEEbiography}

\end{tac}

\end{document}